\documentclass[smallcondensed]{svjour3}
\usepackage[utf8]{inputenc}

\usepackage[round]{natbib}
\usepackage{amsmath}
\usepackage{amssymb}
\usepackage{xcolor}
\usepackage{graphicx}
\usepackage{subcaption}
\usepackage{float}
\usepackage[margin=1.5in]{geometry}
\usepackage{tikz}
\usepackage[colorlinks=true,linkcolor=black,urlcolor=magenta]{hyperref}
\usepackage{enumitem}
\usepackage[normalem]{ulem}

\usepackage{xspace}

\newtheorem{prop}{Proposition}

\newcommand{\oes}{outlying elements\xspace}

\newcommand{\p}{\mathbb{P}}

\begin{document}

\title{Voting Rights, Markov Chains, and Optimization by Short Bursts
}


\author{Sarah Cannon         \and
        Ari Goldbloom-Helzner \and
        Varun Gupta \and
        JN Matthews \and
        Bhushan Suwal
}


\institute{S. Cannon \at
              Claremont McKenna College, 
              Claremont, CA, USA \\
              \email{scannon@cmc.edu}; ORCID: 0000-0001-6510-4669 
           \and
           A. Goldbloom-Helzner \at
              Brown University,
              Providence, RI, USA \\
              \email{ari\_goldblooom-helzner@alumni.brown.edu}
          \and
          V. Gupta \at
              University of Pennsylvania, 
              Philadelphia, PA, USA \\
              \email{vgup@seas.upenn.edu}
         \and
         JN Matthews \at
              MGGG Redistricting Lab, Tisch College, Tufts University, Medford, MA, USA \\
              \email{jenni.matthews@tufts.edu}
         \and
         B. Suwal \at
              Boston University, Boston, MA, USA \\
              \email{bsuwal@bu.edu}
}

\date{Received: date / Accepted: date}

\maketitle

{\bf Abstract. } Finding \oes\ in probability distributions can be a hard problem. 
Taking a real example from Voting Rights Act enforcement, we consider the problem of maximizing the number of simultaneous majority-minority districts in a political districting plan.  
An unbiased random walk on districting plans is unlikely to find plans that approach this maximum. A common search approach is to use a {\bf biased random walk}: preferentially select districting plans with more majority-minority districts.  
Here, we present a third option, called {\bf short bursts}, 
in which an unbiased random walk is performed for a small number of steps (called the {\em burst length}), then re-started from the most extreme plan that was encountered in the last burst.
We give empirical evidence that short-burst runs outperform biased random walks for the problem of maximizing the number of majority-minority districts, and that there are many values of burst length for which we see this improvement. 
Abstracting from our use case, we also consider short bursts where the underlying state space is a line with various probability distributions, and then explore some features of more complicated state spaces and how these impact the effectiveness of short bursts. 



{\bf Acknowledgements} This work began with an idea from Zachary Schutzman, and was explored at the 2019 Voting Rights Data Institute run by the Metric Geometry and Gerrymandering Group. The authors thank the Prof Amar J Bose Research Foundation (MIT) and the Jonathan M Tisch College of Civic Life (Tufts) for their support of the VRDI.  The authors acknowledge valuable contributions from Susy Tovar, Benjamin Gramza, and Moon Duchin. S. Cannon is supported in part by NSF awards DMS-1803325 and CCF-2104795. 
\newpage
\section{Introduction}

The problem of finding small-probability \oes in a large discrete state space is a challenging one. As a motivating example, we consider the problem of finding a political districting plan for a given state or region with a large number of districts that have a majority of their population from a given minority group, such as African-Americans. This is relevant for Voting Rights Act litigation: in launching a claim, plaintiffs typically offer a demonstration that it would have been possible to form a greater number of majority-minority districts than are currently present~\citep{thornburg}. For example, see~\citep{negron,alabama-case,louisiana-case, georgia-case,new-alabama-case}. 
For these reasons, it is an interesting and practically useful problem to find districting plans that maximize the number of majority-minority districts that can be simultaneously created. 

However, the space of possible districting plans is enormous, and one cannot hope to examine every possible plan in any reasonable amount of time. To combat this, mathematicians have turned to random sampling to construct a large collection of districting plans for a given state. This is usually done via a random walk on districting plans: from the current plan, make a random alteration according to specified rules to obtain a new plan, and repeat. Various types of random walk steps have been considered, including a step that flips one geographic unit from one district to another or a step that redraws the boundary between a random pair of districts. 

Rather than using sampling to discover the properties of typical districting plans, we would like to explore districting plans that have as many majority-minority districts as possible.
An unbiased random walk on districting plans will not typically find the most extreme \oes or other highly unlikely events. One widely used approach for these types of problems is a {\it biased} random walk, where one increases the probability of transitioning to a configuration with more majority-minority districts. 
Because a single step changes the number of majority-minority districts by at most one, 
the idea is to make gradual progress toward the \oes we are seeking. 
Biased random walks usually outperform unbiased random walks for finding \oes, but they can still get stuck in local maxima and can fail to find, or even get close to, the desired configurations.  

We present a new method for finding \oes, which outperforms both unbiased and biased random walks in our use case. This method, which we call {\it short bursts}, begins at an arbitrary starting configuration and performs $b$ steps of an unbiased random walk. We then identify the observed configuration with the most extreme value (e.g., the districting plan with the most majority-minority districts) and we re-start another unbiased random walk of length $b$ from there.  Each unbiased random walk of length $b$ is called a short burst, and these short bursts and restarts are performed repeatedly.  An illustration of this process, projected onto the one-dimensional {\it score} being maximized, is shown in Figure~\ref{fig:many_bursts} for bursts of length~$b=5$. 

While we demonstrate the effectiveness of this method for the problem of maximizing majority-minority districts in the context of the Voting Rights Act, the idea itself is much more general. In any application where biased random walks are currently used to find \oes, short bursts has the potential to provide a significant improvement. 

\subsection{Summary of Findings}

After presenting preliminaries in Sections~\ref{sec:background} and~\ref{sec:shortbursts}, we turn to Louisiana's state-level House of Representatives  in 
Section~\ref{section:empirical}. Louisiana was chosen because of its significant Black population and large number of state representatives, which allows for better resolution in our results. We also note that majority-minority districts were a source of contention in the current districting plan for the Louisiana House of Representatives when it was introduced in 2011~\citep{anderson-la}, and Louisiana's U.S. Congressional districts have been challenged in court under the Voting Rights Act as racially gerrymandered~\citep{louisiana-case}. We empirically compare random walks with varying bias to short bursts with varying burst lengths. We see short bursts of almost any length outperform all biased random walks, and that a wide range of burst lengths perform similarly well (Figure~\ref{fig:la-maxes-range}). In fact, the biased random walks on average do not even find a single plan with as many majority-minority districts as the current enacted plan in the Louisiana House of Representatives. This shows that short bursts can be a valuable empirical tool in Voting Rights Act litigation.

We next investigate the reasons for the success of short bursts and look to better understand why short bursts with such a variety of burst lengths perform well.
In Section~\ref{sec:1D}, we consider simpler models and attempt to address these questions. For random walks on a uniform one-dimensional line, short bursts and biased random walks are essentially equivalent. However, when the probability distribution on that line is changed to be approximately normally distributed, short bursts begin to significantly outperform biased random walks, as we see in our Louisiana experiments. However, for this normal model, the shortest burst lengths perform better than all other burst lengths, which is not what we observe in Louisiana where a wide range of short burst lengths performs similarly well. Finally, we build a model of a state space with a bottleneck: in this setting, short bursts of a moderate length are best. Though the findings are preliminary, this suggests there may be underlying bottlenecks in the state space of possible districting plans, and this structure is why we see a wide range of effective short burst lengths rather than shorter bursts always performing better.

\subsection{Related Work}

Most mathematical work on redistricting thus far has focused on randomly sampling typical districting plans, rather than looking for \oes. These typical districting plans can be used to construct a baseline; if a given plan is an outlier with respect to this baseline, this is evidence it is a partisan or racial gerrymander. In the peer-reviewed literature, work in this direction includes~\citep{chen-florida-paper,Chikina2017,Chikina2019,recom,Herschlag2018nc,Carter2019mergesplit}. The process of random sampling  to find the typical properties of legal districting plans has also been used in technical reports~\citep{virginia,santaclara,lowell,chicago} and as evidence in court cases~\citep{pegden2017expert,duchin2017expert,mattingly2017expert,duchin2019brief}.

Though these approaches seek to understand typical properties, there are also numerous applications for finding districting plans that are extreme according to some statistic. Many states have laws that direct district-drawers to favor  plans with certain properties.
For example, North Carolina requires Congressional districting plans to split as few counties as possible, and when sampling districting plans this needs to be considered, as in~\citep{Herschlag2018nc}. Several states have introduced rules requiring districting plans to favor competitive districts, and recent work has explored what happens when districting plans are optimized for competitiveness~\citep{competitiveness}. 

Though extreme plans are sometimes sought in accordance with law and in good faith, it is also clear that extreme plans are sometimes desired for less upstanding purposes, such as to  supply an extreme partisan gerrymander. 
Since this paper presents an effective optimization protocol, we recognize that it could be used to gerrymander.  
However, we are confident that extreme plans found by an optimization protocol would be correctly detected as outliers in a neutral ensemble constructed only according to the enacted rules.

Markov chains are most commonly used to generate likely or typical states, though there has been some work using Markov chains to find \oes. 
For example, in some scientific applications, \oes in high-dimensional distributions can be found by constructing a Markov chain on the data points, calculating the stationary distribution of the Markov chain, and concluding that the data points with low stationary probability are \oes, such as in~\citep{OutliersviaRWs1}. 
To our knowledge, approaches resembling the short burst method have not yet appeared in the scientific literature.

\section{Background} \label{sec:background}

In this section we present relevant background on Markov chains, including how they are used for the purposes of redistricting. 

\subsection{Markov Chains}

The random walks and biased random walks described above are examples of a more general mathematical structure known as a {\it Markov chain}. 
A {\it Markov chain} is a memoryless random process on a state space $\Omega$; we only consider discrete state spaces, where $\Omega$ is all districting plans for a given state or region. 
In particular, a Markov chain randomly transitions between the states of $\Omega$ in a time-independent, or {\it stochastic}, fashion: there are fixed rules that prescribe the probabilities with which the chain transitions to its next state that depend only on its current state.  The probabilities of moves have no dependence on any past behavior of the Markov chain, how long the chain has been running for, or any other factors. For a more detailed background on Markov chains than is included here, see~\citep{lpw}. 

The most common use of Markov chains is to generate random samples via a process known as Markov Chain Monte Carlo. Provided a Markov chain is finite and aperiodic, after enough random steps it converges to a distribution on $\Omega$ known as a {\it stationary distribution}. If one runs a Markov chain for sufficiently long, the current state is approximately a random sample drawn from this stationary distribution. In particular, when $\Omega$ is extremely large, a random sample can often be obtained via this method fairly quickly without needing to enumerate $\Omega$.  As we discuss in more detail in Section~\ref{sec:mc-app}, Markov chains are used to generate collections of random districting plans and detect gerrymandering. 

An elementary example of a Markov chain, which will be of particular relevance in Section~\ref{sec:1D}, is the simple random walk on the line. Here $\Omega = \mathbb{Z}$, the initial state is $X_0 = 0$, and in each time step the chain moves right (to a larger integer) with probability $1/2$ and left with probability $1/2$, no matter the current position. Conditioned on $X_t = a$, this means the state $X_{t+1}$ at time $t+1$ satisfies
\[
\p(X_{t+1} = a+1) = 0.5, \ \  \p(X_{t+1} = a-1) = 0.5,\ \ \p(X_{t+1} = b) = 0 \text{ if } b \neq a+1, a-1
\]
We also consider biased random walks on the line. Like the simple random walk, the initial state is $X_0 = 0$ and the probability of moving right is the same no matter the current position. Unlike the simple random walk, the probabilities of moving left and right are not equal. 
For a biased random walk on the line with bias $p \in (0,1)$, conditioned on the state at time $t$ satisfying $X_t = a$, the state $X_{t+1}$ at time $t+1$ satisfies
\[
\p(X_{t+1} = a+1) = p, \ \  \p(X_{t+1} = a-1) = 1-p,\ \ \p(X_{t+1} = b) = 0 \text{ if } b \neq a+1, a-1
\]
In general, when considering biased random walks we assume $p > 1/2$, that is, the walk is biased rightwards, toward larger integers.

\subsection{Applications of Markov chains in redistricting}\label{sec:mc-app}

The number of possible valid districting plans in any given state is incredibly large.  Consider Minnesota, which has 134 State House districts that are paired to make their 67 State Senate districts: there are 6,156,723,718,225,577,984 ways just to pair State House districts in Minnesota to form State Senate districts, which doesn't even begin to address the number of valid State House districting plans one could draw \citep{alaska}.  We would like to be able to explore the full space of possible districting plans in a state, but we can't do that when it is infeasible to even count the size of that space.  

In redistricting, Markov Chain Monte Carlo (MCMC) methods allow us to generate an \emph{ensemble} of valid plans. This ensemble contains samples from the full space of possible districting plans, which are treated as representative of all possible districting plans.

Political districting plans partition a state or region into connected pieces known as districts. Districts are typically built out of census blocks, the smallest unit of census geography. (There are a few exceptions where districts split blocks, but these are negligibly rare.) In some cases, larger collections of census blocks, like precincts or counties, may be used as the atoms of a plan. Once the basic building blocks are selected, one can form a {\em dual graph} where each node is a geographic unit and the edges indicate geographic adjacency.  A districting plan then is thought of as a partition of the nodes of the dual graph into $d$ connected parts with approximately equal population in each.

To generate a collection of valid districting plans known as an {\it ensemble}, one takes Markov steps from a \emph{seed plan} and collects all plans visited, or some subset of the plans visited, to form the ensemble. Different ways of selecting seed plans and different types of Markov chain steps have been considered~\citep{virginia, Herschlag2018nc, Carter2019mergesplit, recom}. Here, we use Markov chain steps known as Recombination starting from initial plans generated by a recursive spanning-tree method, as in~\citep{recom}. 

These ensembles of typical districting plans have been used to conduct \emph{outlier analysis} in numerous legal challenges to partisan gerrymandering, including in expert reports submitted for court cases in Pennsylvania~\citep{pegden2017expert,duchin2017expert}, Wisconsin~\citep{chen-wisconsin-expert}, Florida~\citep{chen-florida-report}, and North Carolina~\citep{chen-nc, mattingly2017expert, duchin2019brief}. In outlier analysis, the current or proposed plan is compared to the constructed ensemble of typical districting plans along certain statistics of interest such as number of seats won by each party. It is deemed to be gerrymandered if it is an extreme outlier compared to the ensemble. 

While MCMC approaches are most commonly used to generate ensembles against which a proposed plan could be flagged as an outlier, sometimes we are interested in generating such extreme plans themselves. One powerful example is plans that have more majority-minority districts than in the enacted plan. In the landmark U.S. Supreme Court case \cite{thornburg}, the court established the framework upon which to build Voting Rights Act cases against racial gerrymandering.  One of those threshold conditions is that \emph{``... the minority group must be able to demonstrate that it is sufficiently large and geographically compact to constitute a majority in a single-member district.''} 
While the VRA does not ultimately require the creation of majority-minority districts, to launch a claim, plaintiffs typically demonstrate that \textit{it would have been possible} to form a greater number of such districts than are present in the system under challenge.\footnote{In the important case {\em Bartlett v. Strickland} (2009), the court reaffirmed the importance of districts with over 50\% minority population share, as opposed to districts in which the minority group could merely exert influence.}  This demonstration step is an important application of optimization techniques in practical redistricting. Markov chains like Recombination are unlikely to find these \oes in a reasonable amount of time simply because of the size of the sample space of districting plans, especially as they are designed for representative rather than extreme sampling. Biasing runs towards accepting more extreme plans gives an improvement, but we present another method which we find is empirically even better at finding extreme plans, such as those with many majority-minority districts. 

\section{Short Bursts} \label{sec:shortbursts}

We begin by formalizing the problem. Let $\Omega$ be a discrete state space, and let $\mathcal{M}$ be a Markov chain on $\Omega$.  Let $s: \Omega \rightarrow \mathbb{Z}$ be a score function where if $\sigma$ and $\tau$ are two states of $\Omega$ that differ by one transition of $\mathcal{M}$, then $|s(\sigma) - s(\tau)| \leq 1$. The goal is to find a state with as large a score as possible. We do not expect our results to change for real-valued score functions or other small, constant bounds on the score differences between adjacent states, provided these differences are small with respect to the overall range of possible scores. 
In our applications, $\Omega$ is all districting plans for a state, $\mathcal{M}$ is the Recombination Markov chain, and $s(\sigma)$ is the number of majority-minority districts in districting plan $\sigma$. 
 
One approach to finding the state in $\Omega$ with the largest score is to simply list everything in $\Omega$ and calculate the score of each. For many applications, including ours, this is completely infeasible due to the large size of $\Omega$. We note that because we cannot hope to examine all of $\Omega$, or even a sizeable fraction of the states in $\Omega$, we are unable to find the state with the maximum score, or to certify that a state with a large score is in fact the maximum.  When we talk below about maximizing score functions, we mean finding a state with as large a score as we can find, not finding the true maximum score. 

We propose a new method for maximizing the score function which we call {\it short bursts}. This process works in phases known as {\it bursts}. In one burst, $b$ steps of Markov chain $\mathcal{M}$ are performed from the starting state; $b$ is called the {\it burst length}.  The states visited in the burst are examined, and the state $\sigma$ with the largest score is found. If there are multiple states with the same score visited during a burst, one should pick the most recently visited.  The next burst begins from $\sigma$, and this process is repeated as many times as is desired.  See Figure~\ref{fig:many_bursts}, which depicts short bursts of length $5$ for the simple random walk on the line where the score of a state is its integer label. For more complicated state spaces, the projection of short bursts in $\Omega$ onto the score functions looks similar, though the score may at times stay the same for multiple iterations. In addition to burst length, sometimes it will be more convenient to discuss {\it burst size}, which is the total number of plans visited during a short burst: since this counts both the starting and ending plans, it is one more than the burst length. 

\newcommand{\xMin}{-2}
\newcommand{\xMax}{5}
\newcommand{\yMin}{0}
\newcommand{\yMax}{8}
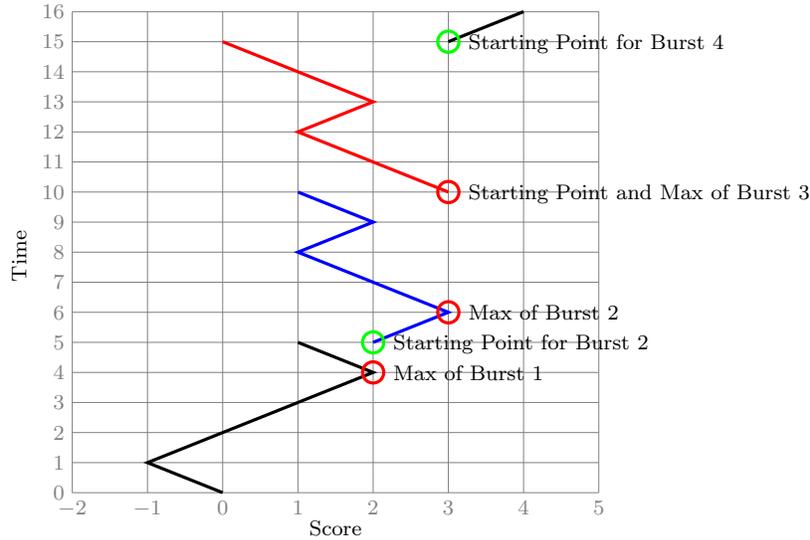
\begin{figure}
    \centering

\begin{tikzpicture}[yscale=0.8]
    \foreach \i in {\xMin,...,\xMax} {
        \draw [very thin,gray] (\i,\yMin) -- (\i,\yMax)  node [below] at (\i,\yMin) {$\i$};
    }
    \foreach \i in {\yMin,...,16} {
        \draw [very thin,gray] (\xMin,\i/2) -- (\xMax,\i/2) node [left] at (\xMin,\i/2) {$\i$};
    }
\draw[very thick,black] (0,0) -- (-1,.5) -- (1,1.5) -- (2,2) -- (1,2.5);
\draw[very thick, blue] (2,2.5) -- (3,3) -- (2,3.5) -- (1,4) -- (2,4.5) --(1,5);
\draw[very thick, red] (3,5) -- (1,6) -- (2,6.5) -- (0,7.5)  ;
\draw[very thick, black] (3,7.5) -- (4,8);
\draw 
node[label={[label distance=0cm,text depth=-2ex,rotate=90]left:Time}] at (-2.6, 4.5) {}
node[label={[label distance=0cm,text depth=-2ex]below:Score}] at (1.5, -.2) {}
node[circle, red,very thick, draw, label={[label distance=0cm,text depth=-3ex]right:Max of Burst 1}] at (2,2) {}
node[circle, green, very thick, draw, label={[label distance=0cm,text depth=-2ex]right:Starting Point for Burst 2}] at (2,2.5) {}
node[circle, green, very thick, draw, label={[label distance=0cm,text depth=-2ex]right:Starting Point for Burst 4}] at (3,7.5) {}
node[circle, red, very thick, draw, label={[label distance=0cm,text depth=-2ex]right:Max of Burst 2}] at (3,3) {}
node[circle, red, very thick, draw, label={[label distance=0cm,text depth=-2ex]right:Starting Point and Max of Burst 3}] at (3,5) {};
\end{tikzpicture}
    \caption{The position over time for short bursts of length five for the simple random walk on the line, with score equal to position. After each five steps of this simple random walk, we re-start at the highest position visited so far, take five more steps, and repeat.}
    \label{fig:many_bursts}
\end{figure}

We will revisit short bursts for the simple random walk on the line and give a precise mathematical characterization of it in Section~\ref{sec:1D}.  First, we demonstrate empirically that, somewhat surprisingly, for the redistricting application of maximizing majority-minority districts, short bursts of nearly all considered lengths outperform biased random walks for finding districting plans with large scores.

\section{Majority-Minority Districts in Louisiana}
\label{section:empirical}

In this section we show that short bursts outperform biased random walks for finding districting plans for the Louisiana  House of Representatives with a large number of majority-minority districts. In addition to Louisiana, we also performed this analysis on New Mexico, Texas, and Virginia; our findings for Texas are discussed in Section~\ref{subsec:coalitions} in the context of voting coalitions, while our results for New Mexico and Virginia are summarized in Appendix~\ref{app}. 

We saw similar results in all cases---with short bursts outperforming biased runs---in states where the minority population accounted for roughly 30\% or more of the voting age population. For states where the minority population accounts for less than 20\% of the total voting age population, there were similar trends but experimental results were too tightly clustered to obtain conclusive findings. 

We begin with some fundamentals about districting in Louisiana, and then describe our methods in more details before giving our results, which indicate that short bursts clearly outperform both unbiased and biased random walks. 

\subsection{Louisiana: Demographics and Districting}

We focus our analysis on Louisiana as it is a state with both a significant People of Color population (around 30\% Black and 10\% other People of Color populations) and a large enough number of State House seats (105). The latter provides us with good resolution when our results are projected onto numbers of seats. At the same time, Louisiana is populous enough that population balance between districts can be achieved even with districts drawn on census block groups.

\subsubsection{Minority Representation in the Louisiana State Legislature}

The Louisiana State Legislature is comprised of the Louisiana Senate, with 39 seats, and the Louisiana House of Representatives, with 105 seats.
In the 2010 Decennial Census, the non-Hispanic Black population (BPOP) made up 31.82\% of Louisiana's total population (TOTPOP), and the non-Hispanic Black voting age population (BVAP) made up 29.85\% of Louisiana's voting age population (VAP).  Table \ref{tab:LA_demos} details demographic information from the 2010 Decennial Census as well as the 2014-2018 American Community Survey (ACS) Estimates.  (Note that here and throughout this section, ``Hispanic" is used to refer to all Census/ACS respondents who selected Hispanic or Latino as their ethnicity regardless of their race selection, ``White" refers to non-Hispanic White alone, ``Black" refers to non-Hispanic Black alone, and so forth.)

\begin{table}
    \centering
    \small{
    \begin{tabular}{|p{4.1cm}||c|c|c|}
    \hline
        &2010 Census & 2010 Census & 2014-2018 ACS \\
            & TOTPOP & VAP & TOTPOP     \\
    \hline
        White       & 2,734,884 (60.33\%) & 2,147,661 (62.88\%) & 2,744,265 (58.84\%) \\
        Black       & 1,442,420 (31.82\%) & 1,019,582 (29.85\%) & 1,492,230 (32.00\%) \\
        Hispanic    & 192,560 (4.248\%) & 138,091 (4.043\%)   & 234,920 (5.037\%) \\
        Asian       & 69,327 (1.529\%) & 53,638 (1.57\%)    & 79,137 (1.697\%) \\
        Two or more races   & 57,766 (1.274\%) & 30,755 (0.9005\%) & 78,991 (1.694\%) \\
        Amer. Indian/Alaska Native & 28,092 (0.6197\%) &  19,952 (0.5842\%) & 24,014 (0.5149\%)\\
        Other races & 6,779 (0.1495\%) & 4,526 (0.1325\%)    &  8,905 (0.1909\%) \\
        Native Hawaiian and Other  \newline \textcolor{white}{.} \hspace{1mm} Pacific Islander  & 1,544 (0.0341\%) & 1,152 (0.0337\%) & 1,154 (0.0247\%)\\
    \hline
        Total & 4,533,372 & 3,415,357 & 4,663,616 \\
    \hline
    \end{tabular}}
    \caption{Racial demographics in Louisiana. TOTPOP is total population, VAP is Voting Age Population, and ACS refers to the Census Bureau's American Community Survey. }
    \label{tab:LA_demos}
\end{table}

\begin{figure}
   \centering
    \includegraphics[height=0.25\textheight]{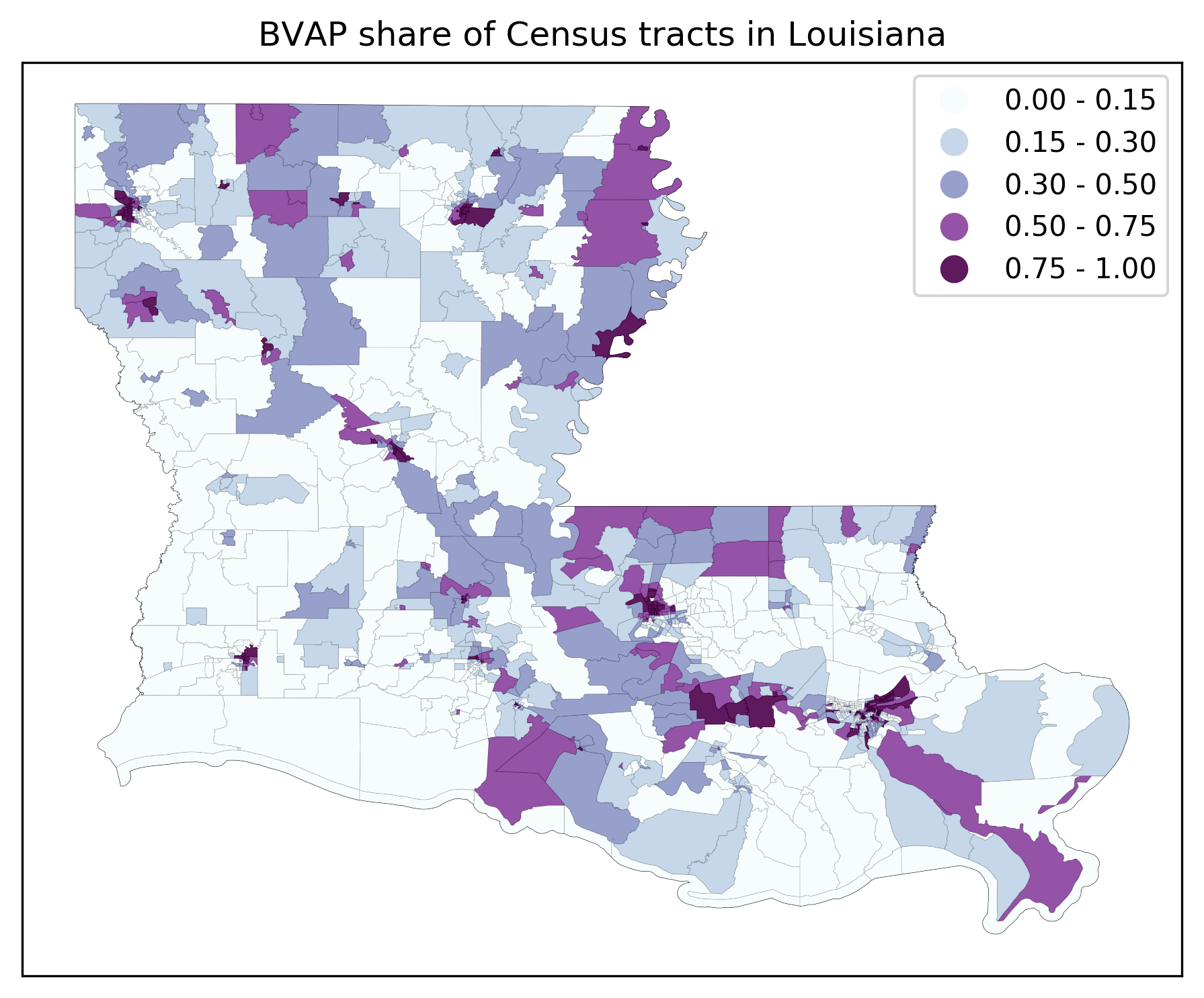}
    \qquad
    \includegraphics[height=0.25\textheight]{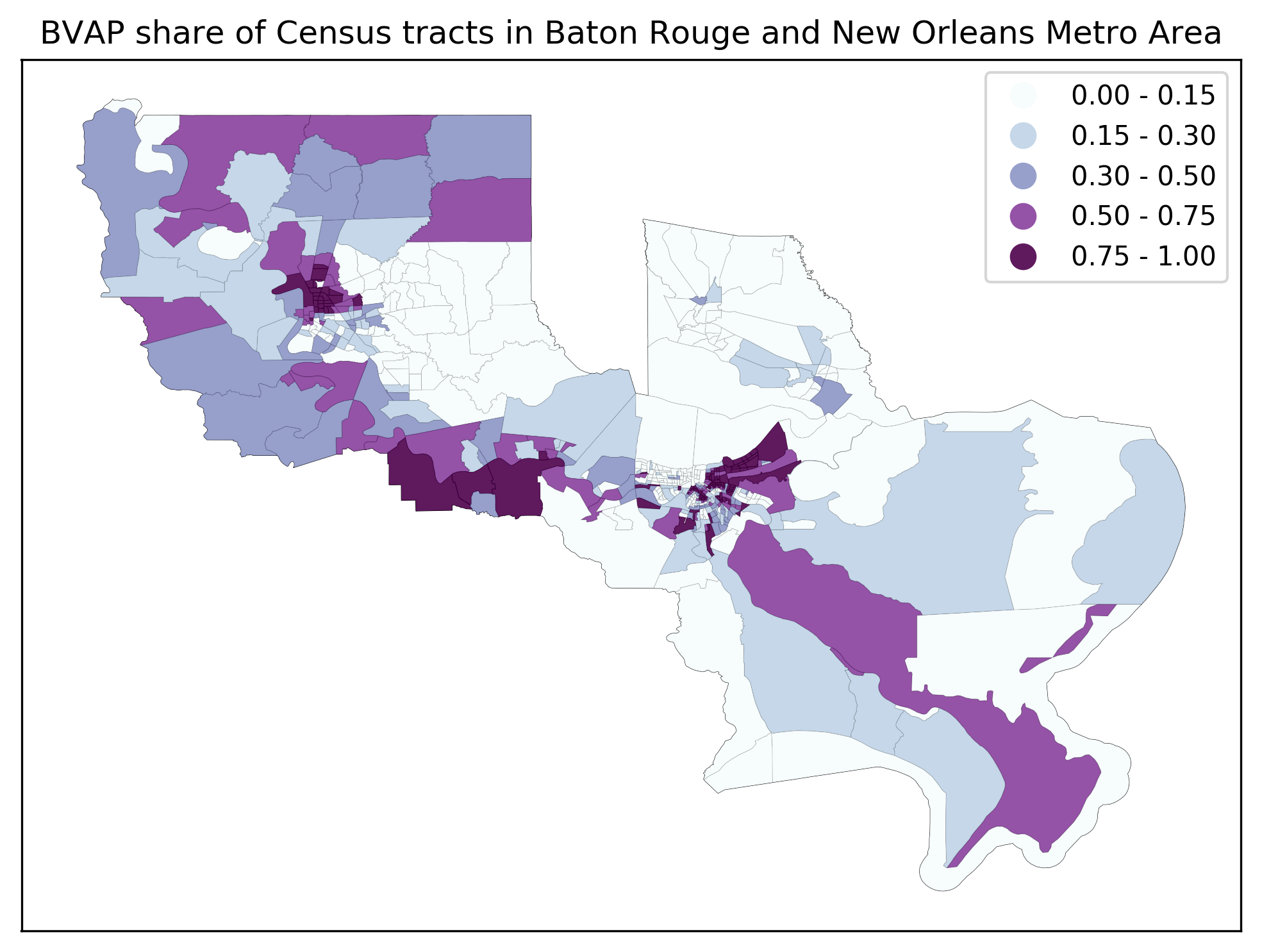}
    
    \caption{Distribution of Black voting age population in Louisiana.  The top figure shows the whole state and the bottom figure shows the combined metropolitan areas of Baton Rouge and New Orleans.}
    \label{fig:la_choro}
\end{figure}

Of Louisiana's 39 senators, 10 are members of the Louisiana Legislative Black Caucus (LLBC); in the Louisiana House of representatives, 27 of the 105 representatives are members of the LLBC~\citep{llbc}.
 From 2010-2018, there was one Hispanic representative in the State House, Helena Moreno, who resigned from the House of Representatives after being elected as Councilmember At-Large on the New Orleans City Council.  The Louisiana House of Representatives currently reports no Hispanic members.

The Black voting age population exceeds 50\% of the voting age population
in 10 State Senate districts and 28 State House districts. 
There are currently no State House or State Senate districts where the Hispanic voting age population share is over 50\%.  The House district with the largest such share is District 93, where the Hispanic population is 21.1\% of the voting age population.

\subsubsection{Redistricting in Louisiana}

Like all states, Louisiana must comply with federal equal population requirements and Section 2 of the Voting Rights Act.  Additional guidelines were adopted by the state legislative committees tasked with redistricting~\citep{LAHouseRules, LASenateRules}. These guidelines call for districts to be contiguous; that districts should respect recognized political boundaries and natural geography of this state, to the extent practicable; and that due consideration should be given to existing district alignments.  In addition they state that no plan with ``absolute deviation of population which exceeds plus or minus five percent of the ideal district population" shall be considered for the Louisiana Senate or the Louisiana House of Representatives.

\subsection{Methods}

In order to most prominently see the effects of short bursts, we chose states that had both a relatively large number of state House districts and a large minority population. While the results we first present here focus on BVAP in Louisiana, we also used similar methods to compare short bursts and biased runs for People of Color voting age population in Louisiana and for various minority populations in New Mexico, Texas, and Virginia; see Section~\ref{subsec:coalitions} and Appendix~\ref{app}.

\subsubsection{Data}

2010 Decennial Census demographic data was downloaded from the Census API. The 2010 census block group shapefile for Louisiana, as well as block and tract level shapefiles and state legislative districts, is available from the US Census Bureau’s TIGER/Line Shapefiles.
Block groups were assigned to state legislative districts using MGGG's geospatial preprocessing package \url{maup}, available at \href{https://github.com/mggg/maup}{https://github.com/mggg/maup}. 

Our experiments focus on the block group level data, although data at other levels is included in our state level analysis.
The final data for Louisiana (and other states) can be found in the github repository %
\url{https://github.com/vrdi/shortbursts-gingles/tree/master/state_experiments}.

\subsubsection{Experiments}

We want to understand how short bursts performs in maximizing the number of majority-minority districts in Louisiana. Specifically, the score we assign to each districting plan is the number of districts where the voting age population is over $50\%$ black (majority BVAP districts). For each short burst length $b \in \{2,5,10,25,50,100,200 \}$, we ran ten trials, keeping track of the largest score seen among districting plans visited; for further granularity, the additional burst lengths $b \in \{6,7,8,9,11,12,13,14,15,20\}$ were then considered as well.  For comparison, we also performed some biased and unbiased random walks that did not use the short bursts method.  For our biased random walks, if a new proposed plan increased the score or kept the score the same, we always accepted it; if it decreased the score, we accepted it with probability $q$. We ten trials for each $q \in \{ \frac{1}{4}, \frac{1}{8}, \frac{1}{16}\}$, where $q = 1$ corresponds to an unbiased random walk. Here smaller $q$ values correspond to a heavier bias in the random walk.

In each of these trials we used block groups as our basic geographic building blocks for districts. Population constraints were set so that each State House district can be at most 4.5\% away from the ideal population. This was the tightest population bound for which we could find a good starting plan, and is a slightly tighter bound than the level of population deviation present in Louisiana's currently enacted plan.

Each trial is started at a \textit{seed plan} that we've generated with a recursive spanning tree method  to be within the population constraints (for more details on this method, see~\citep{virginia}). We then generate 100,000 additional random valid districting plans beginning at this seed plan using the Recombination Markov chain developed by MGGG~\citep{recom}. At each step, in addition to checking a proposed plan has balanced population, we check a compactness constraint, ensuring  the number of dual graph edges whose endpoints are in different districts is no more than twice the number of such edges present in the seed plan.

\subsection{Results}

\subsubsection{Unbiased Ensemble Baseline}

\begin{figure}
    \centering
    \includegraphics[width=0.75\textwidth]{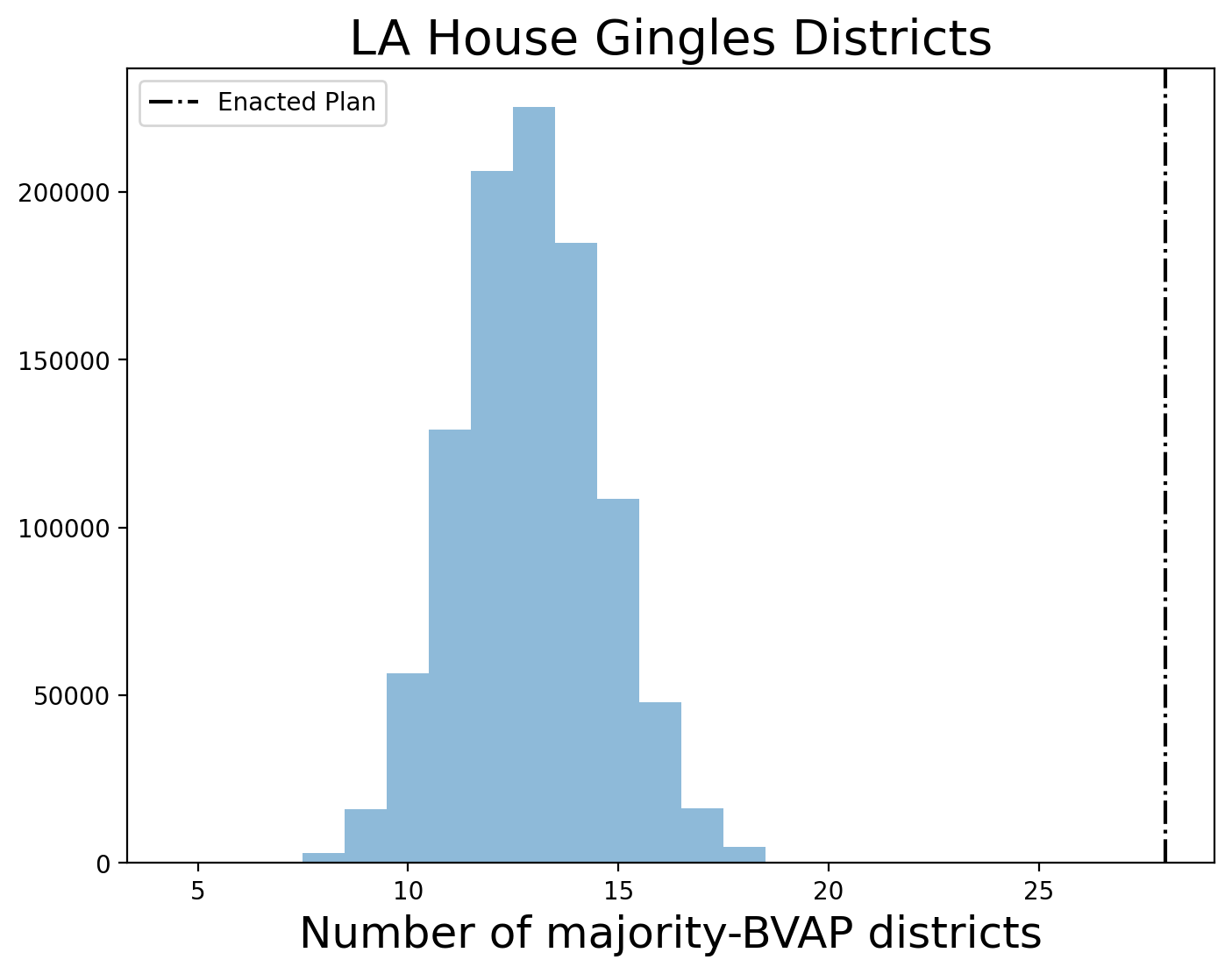}
    \qquad
    \includegraphics[width=\textwidth]{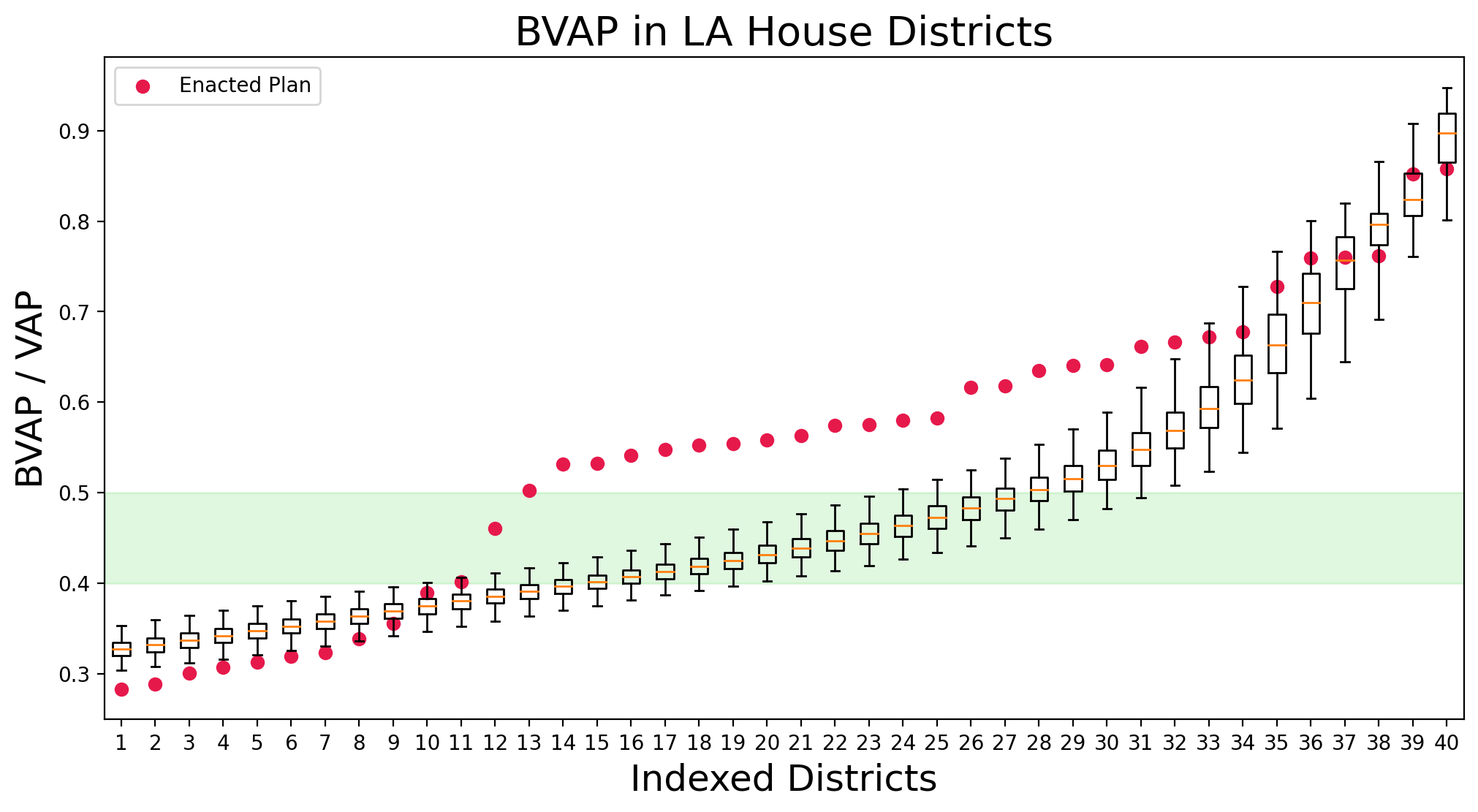}

    \caption{The distribution of majority-BVAP districts in randomly sampled districting plans for the Louisiana state House of Representatives. The histogram (top) shows the number of sampled plans with various values for the number of majority-BVAP districts.  The boxplot (bottom) shows the observed distribution of BVAP percentage in the top 40 State House districts indexed by BVAP. That is, the rightmost box shows the BVAP percentage in the district with the largest BVAP percentage across all sampled plans, the second-rightmost box shows the BVAP percentage in the district with the second-largest BVAP percentage across all sampled plans, etc.
    The enacted plan has 28 majority-BVAP districts in the House, indicated by a dashed-dotted line in the histogram and the red dots on the boxplot.  Of these 28 districts, 27 have elected Black representatives.
    }
    \label{fig:unbiased_results}
\end{figure}

First, we ran an unbiased random walk on the state space to understand some of the baseline properties of Louisiana State House districting plans; see Fig.~\ref{fig:unbiased_results} which shows a histogram of the number of majority-BVAP districts in each plan encountered and a boxplot of the BVAP in the top 40 districts with the largest BVAP. In particular, we se typical plans have an average of about 13 majority-BVAP districts, though plans with as many as 18 majority-BVAP districts were found by this unbiased Markov chain. This baseline will serve to help us compare biased random walks and short bursts.

\subsubsection{Comparing Short Bursts and Biased Runs}

In this section, we use both short bursts and biased runs to find districting plans with as many majority-BVAP districts as possible. 
\begin{figure}
    \centering
    \includegraphics[width=0.8\textwidth]{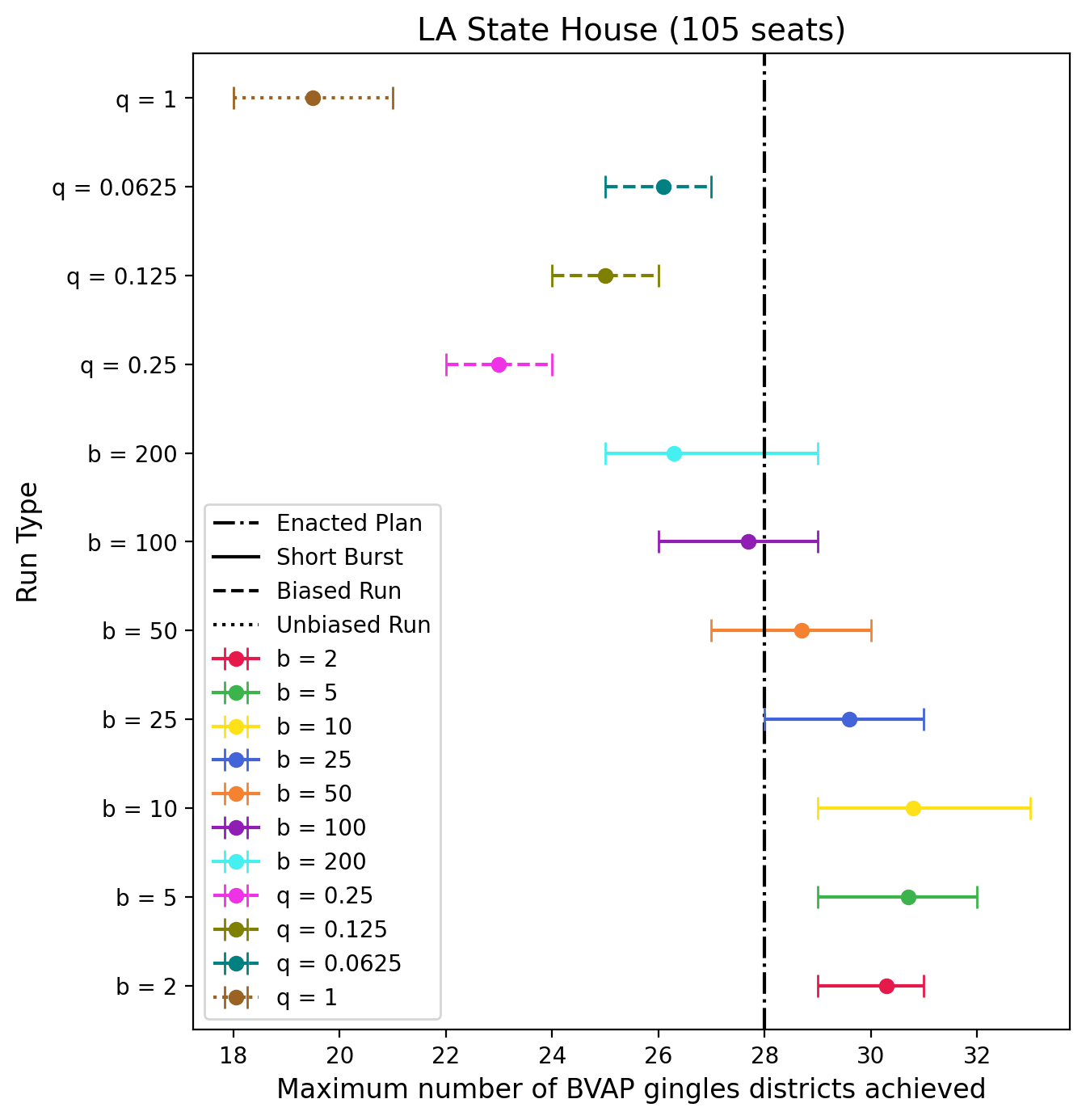}
    \caption{Range of maximum numbers of majority-BVAP districts observed for short bursts, biased runs, and unbiased ($q$ = 1) runs.  The dot is the mean across the ten trials and the bars indicate the min/max range.  Each type of 100,000 step run was performed 10 times.  The short burst runs not only perform better than the biased and unbiased runs, but they also all find a plan with more Gingles districts than the enacted plan.}
    \label{fig:la-maxes-range}
\end{figure}

Figure~\ref{fig:la-maxes-range} shows the maximum number of majority-BVAP districts found using short burst runs with bursts of varying lengths, from 2 to 200; for biased random walks with three different biases $ q = 0.25$, $q = 0.125$, and $q = 0.0625$; and for unbiased runs ($q = 1$).  To be sure that our results were not atypical, this same experiment was performed ten times in each case. 

While the unbiased process, across 10 different runs, only found districting plans with maximum scores between 18 and 21 in 100,000 steps, short bursts of all considered lengths easily found plans with far more than 21 majority-BVAP districts. In fact, short bursts found districting plans with up to 33 majority-BVAP districts. 

A biased run with bias $q$ follows the same steps as an unbiased run, with one exception: If the proposed new plan has fewer majority-BVAP districts than the current plan, then the Markov chain only moves to that new plan with probability $q$ and remains at its current plan otherwise.  Note that a biased run with probability $q$ = 1 is equivalent to an unbiased run.  While the runs with bias $0.0625$ are fairly comparable to short bursts of length $200$, otherwise all short burst runs outperformed all biased runs. Short bursts, when run for the same number of total steps as a biased random walk, were able to find plans with up to six additional majority-BVAP districts, a remarkable improvement. We note that biased runs with bias smaller than $0.0625$ were not considered as rejection probabilities higher than $0.9375$ significantly slow down the chain, leading it to stay at the same districting plan for long periods of time.

The success of the short bursts approach compared to biased random walks is robust to the choice of burst length: every short burst run for $b = 2,5,10,25$ outperformed every biased random walk; the worst short burst run for $b = 50$ had the same performance as the best biased random walk; and short bursts runs with $b = 100$ still noticeably outperformed biased random walks on average. For further granularity, we also considered additional burst lengths between 5 and 25, and saw similarly strong performances; see Figure~\ref{fig:LA-bvap-maxes-more}. 

\begin{figure}
    \centering
    \includegraphics[width = 0.8\textwidth]{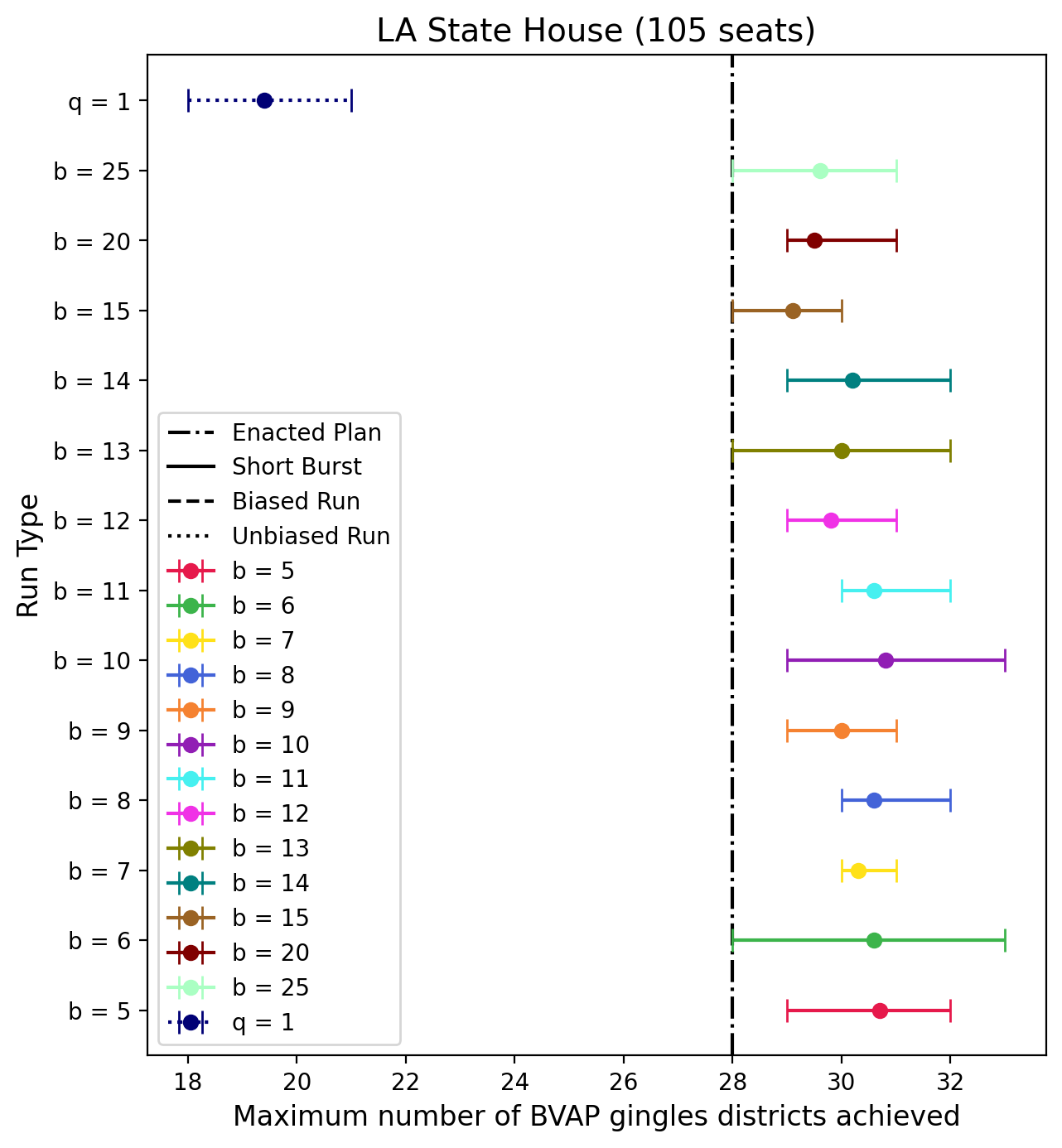}
    \caption{Range of maximum numbers of majority-BVAP districts observed for short bursts of different lengths. Each type of 100,000 step run was performed ten times. The dot is the mean across the ten trials, and the bars indicate the min/max range. All of these burst length, from 5 up to 25, perform similarly well, and better than any biased random walks (see Fig.\ref{fig:la-maxes-range} for the outcomes of similar experiments for biased and unbiased random walks).}
    \label{fig:LA-bvap-maxes-more}
\end{figure}

For practitioners interested in applying this method in political districting contexts, we note that $b = 10$ appears to perform best, with both the highest average score achieved and highest maximum value, and would therefore suggest it as a starting point. We would expect to see little difference if nearby values from 5 to 25 were used.  We do not have statistical evidence that 10 is the best choice, due to the difficulty of distinguishing between the similarly good performances of several burst lengths. We feel this is a strength of our method, as one need not find the exact optimal burst length in order see marked improvements over other methods such as biased random walks.

We also looked at how quickly this separation between the different methods emerged. The average number of majority-BVAP districts found over time for short bursts of the  best observed burst length ($b = 10$), a biased random walk with the best observed bias ($q = 0.0625$), and an unbiased random walk are shown in Fig.~\ref{fig:la-maxes-all-bvap}. The shaded cones around the average lines denote the min/max range across the ten runs. This plot demonstrates the separation between short bursts and bias random walks happens quite early on, by around 10,000 steps, and persists throughout the remaining 90,000 steps. 

\begin{figure}
    \centering
    \includegraphics[width=\textwidth]{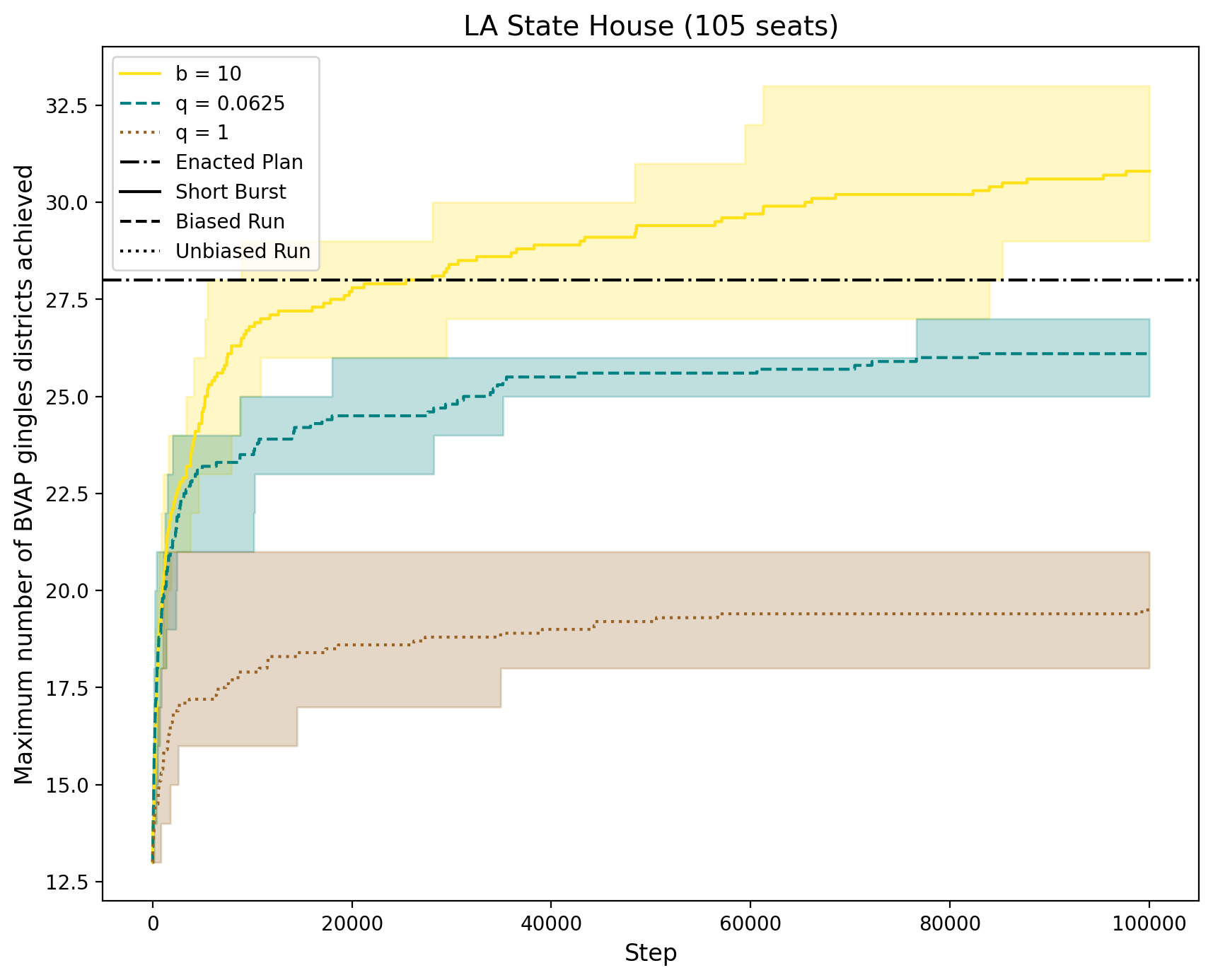}
    
    \caption{Maximum numbers of majority-BVAP districts observed at each step in the chain for the best preforming short burst length ($b$ = 10), biased run ($q$ = 0.0625), and the unbiased run ($q$ = 1).  The line indicates the mean, and the colored band the min/max range, across the trials.  Each type of 100,000 step run was performed 10 times.  
    The short bursts runs outperform the biased runs.
    }
    \label{fig:la-maxes-all-bvap}
\end{figure} 

\subsection{Different Minority Groups and Voting Coalitions} 
\label{subsec:coalitions}

In Louisiana, there is only one minority group with a large enough population to comprise a majority in any district in the state's current districting plan. However, it is a well-established fact that people of color often form voting coalitions to elect minority representatives. For this reason, it may make sense to instead consider the People of Color Voting Age Population (POCVAP) rather than BVAP in Louisiana, as this may give a better representation of how likely a district is to elect a minority representative. We saw similar results to our experiments using BVAP, though in fact there was an even larger separation between the success of biased runs and of short bursts; see Figure~\ref{fig:la-maxes-all-pocvap}. 

\begin{figure}
    \centering
    \begin{subfigure}{0.44\textwidth}
        \centering
        \includegraphics[width=\textwidth]{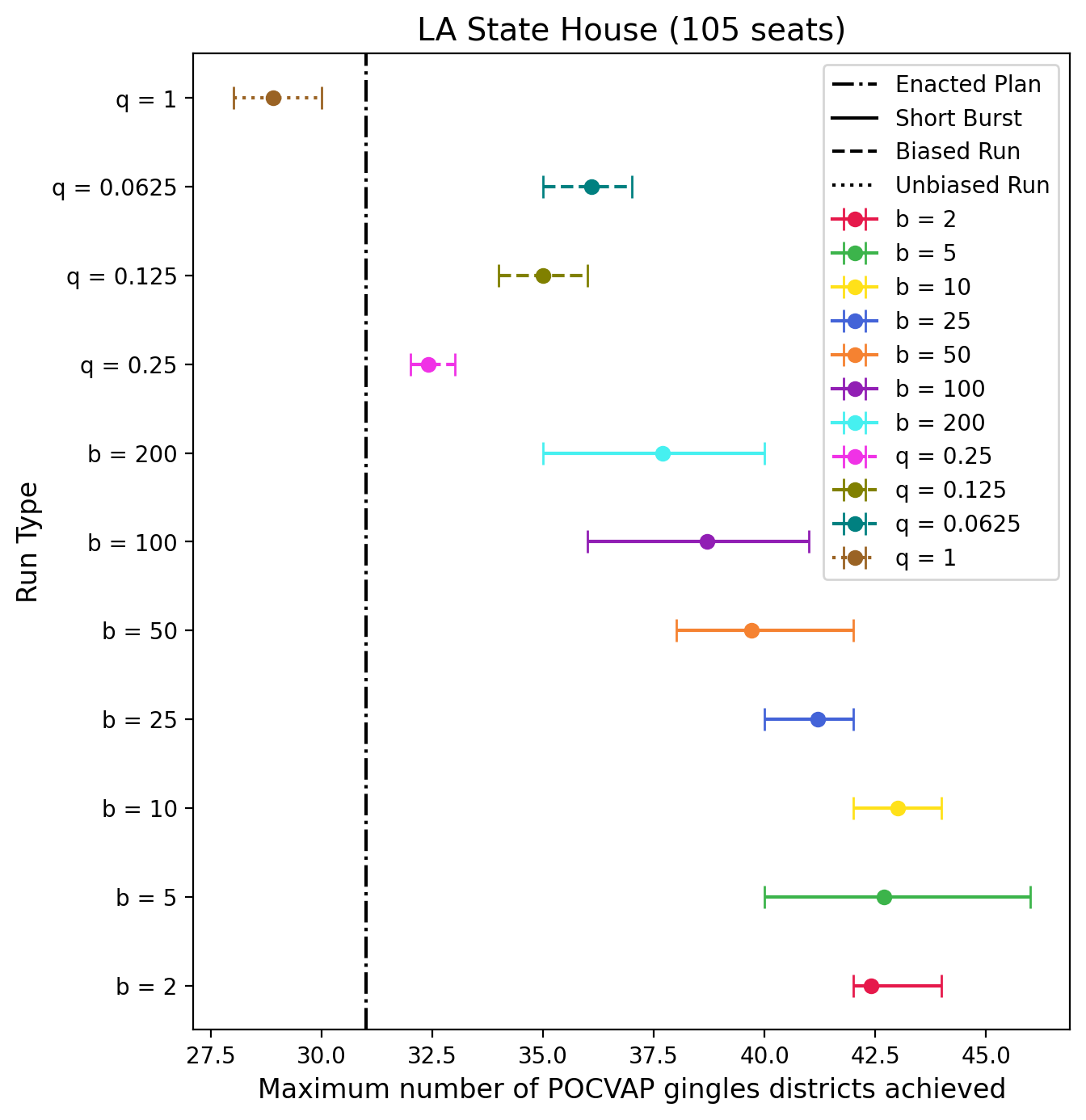}
        \caption{Maximum Gingles districts observed}
    \end{subfigure}
    \begin{subfigure}{0.55\textwidth}
        \centering
        \includegraphics[width=\textwidth]{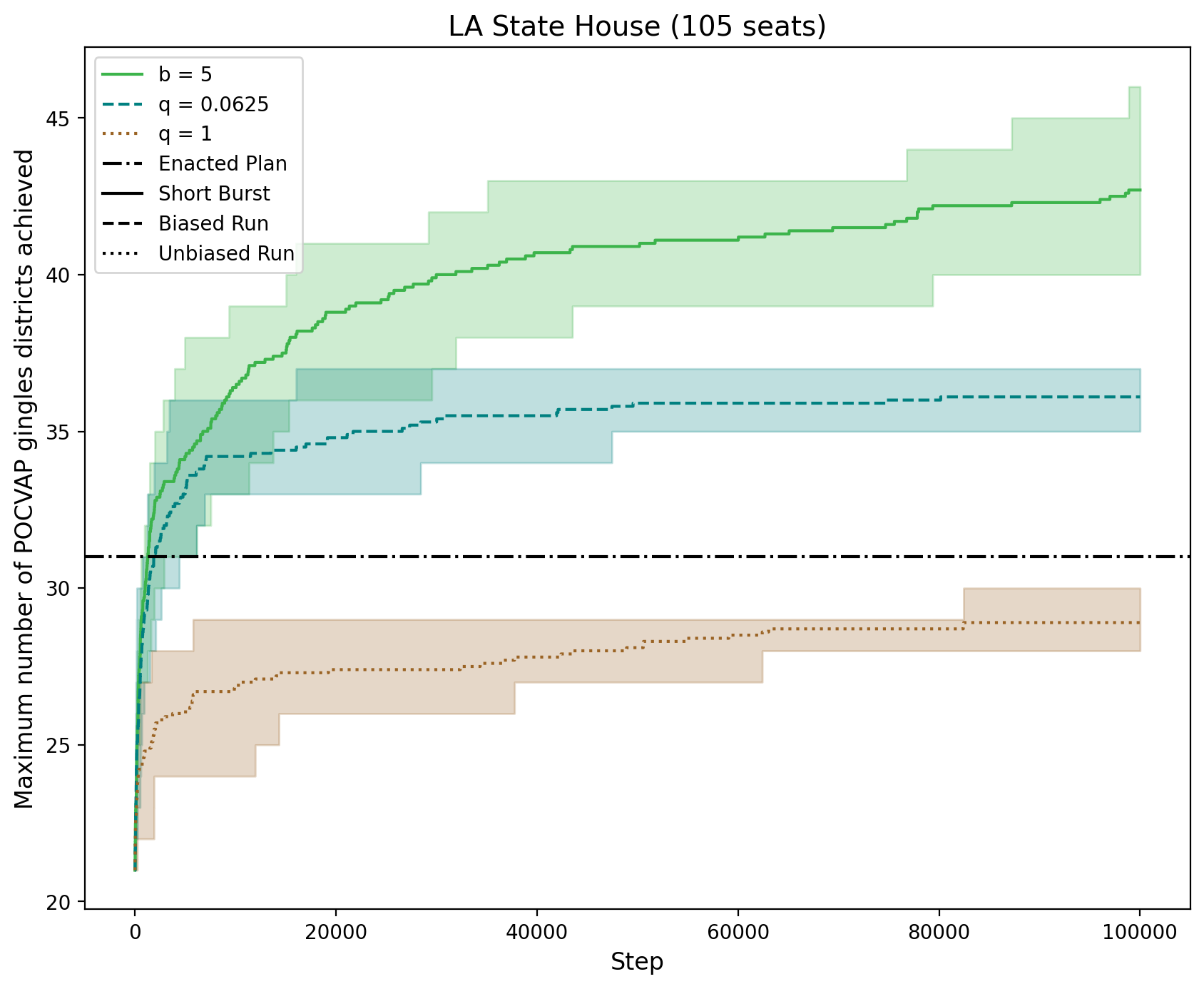}
        \caption{Maximums observed over time}
    \end{subfigure}
    
    \caption{Louisiana | Maximum numbers of majority-POCVAP districts observed for short bursts, biased runs, and an unbiased run ($q$ = 1). Each type of 100,000 step run was performed 10 times.  The left figure (a) shows the range of maxes observed for each run type. The dot is the mean across the ten trials and the bars indicate the min/max range.  The right figure (b) shows maximum number of majority-POCVAP districts achieved at each step in the chain for the best preforming short burst length ($b$ = 5), biased run ($q$ = 0.0625), and the unbiased run ($q$ = 1).  The line indicates the mean, and the colored band the min/max range, across the trials.  The short bursts runs outperform the biased runs.
    }
    \label{fig:la-maxes-all-pocvap}
\end{figure}

The nearby state of Texas has both a large Black and a large Hispanic population. We ran similar experiments as in Louisiana for the 150-seat Texas House of Representatives, but considered the BVAP (Black Voting Age Population, 11.4\% of VAP),  HVAP (Hispanic Voting Age Population, 33.6\% of VAP), and POCVAP (People of Color Voting Age Population, 50.4\% of VAP). We obtained data from similar governmental sources, allowed districts to differ in population from their ideal population by at most 2\%, and began at a random initial plan. 

For maximizing the number of majority-HVAP districts, short bursts clearly outperformed other methods; see Figure~\ref{fig:tx-maxes-all}(a-c). Short bursts also appear to outperform other methods for majority-BVAP districts, though there is too much overlap in the expected maximum curves to be sure; see Figure~\ref{fig:tx-maxes-all}(d-f). This was consistent with observations in other states: when the range of possible values for the number of majority-minority seats is too narrow, it is hard to distinguish between the methods. Texas's range of possible majority-BVAP districts is indeed quite narrow, only ranging from 1 to 6 in our experiments. This narrow range of possible values is likely a factor of both Texas's comparably small black population and the spatial distribution of these voters.

However, it makes sense to consider the Black and Hispanic populations simultaneously, rather than in separate experiments. While one could consider a multi-dimensional process where the score function is, say, the sum of the BVAP Gingles districts and the HVAP Gingles districts, this ignores the fact that Black and Hispanic populations often vote similarly and form voting coalitions. Instead, we consider the POCVAP, the total voting age population consisting of people of color; see Figure~\ref{fig:tx-maxes-all}(g-i). Using POCVAP, rather than some combination of BVAP and HVAP, gives a more accurate representation of the number of districts that have an opportunity to elect a minority candidate. For POCVAP-majority districts in Texas, some short burst experiments were able to find 94 such districts, as compared to maximums of 48 HVAP Gingles districts and 6 BVAP Gingles districts. Short bursts also outperformed other methods of maximizing the total number of POCVAP-majority districts. 
        
We briefly note that for all of these experiments, a burst length of $b = 10$ had the highest average value achieved across 10 runs. However, there was similarly strong performance for other burst lengths from 2-25, so we do not claim with statistical certainty that 10 will always be the best burst length. In all cases, there was a wide range of short burst lengths that outperformed all biased random walks. 

\begin{figure}
    \centering
    \begin{subfigure}{0.9\textwidth}
        \begin{subfigure}{0.44\textwidth}
            \centering
            \includegraphics[width=\textwidth]{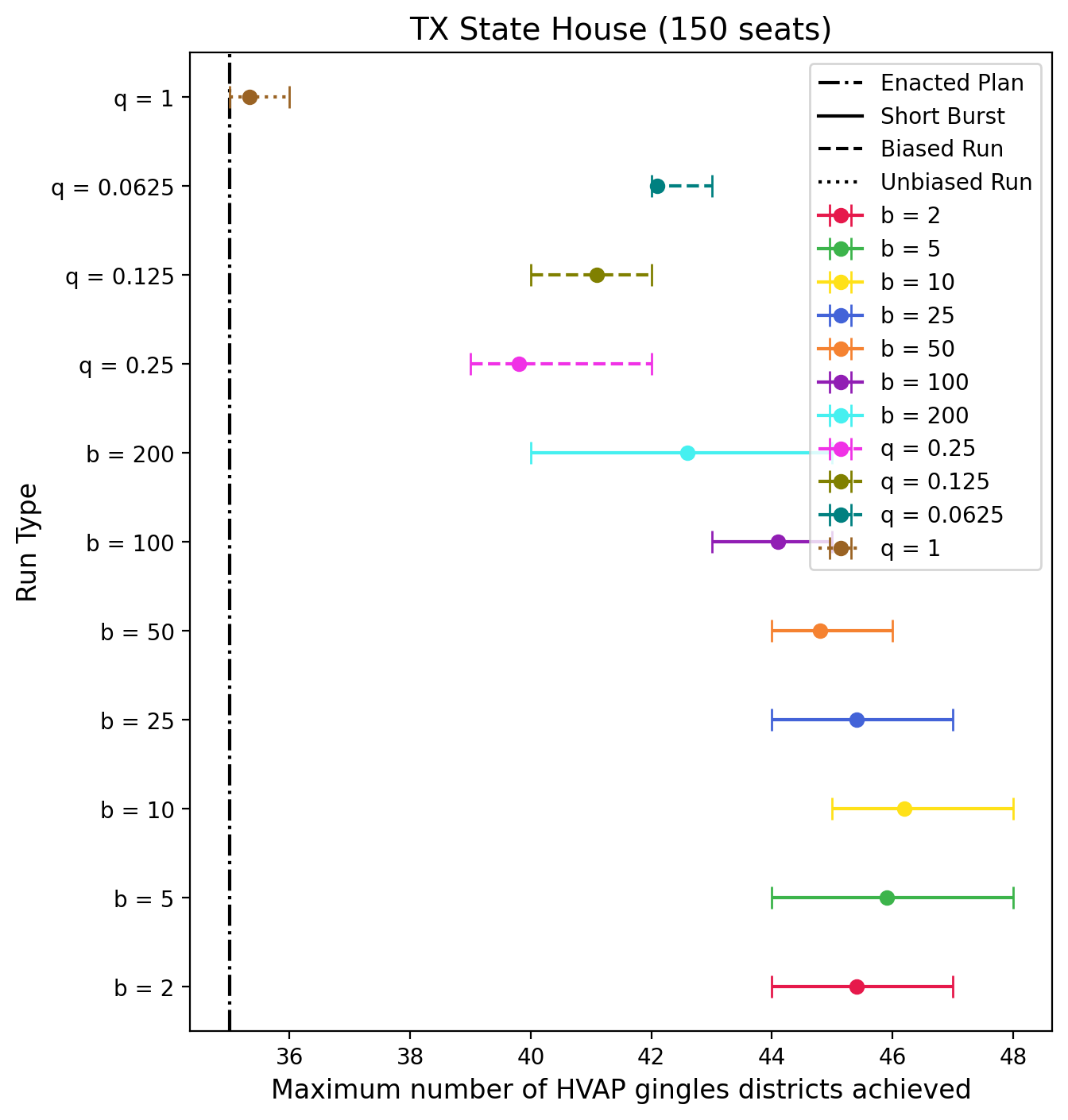}
            \caption{Maximum Gingles districts observed}
        \end{subfigure}
        \begin{subfigure}{0.55\textwidth}
            \centering
            \includegraphics[width=\textwidth]{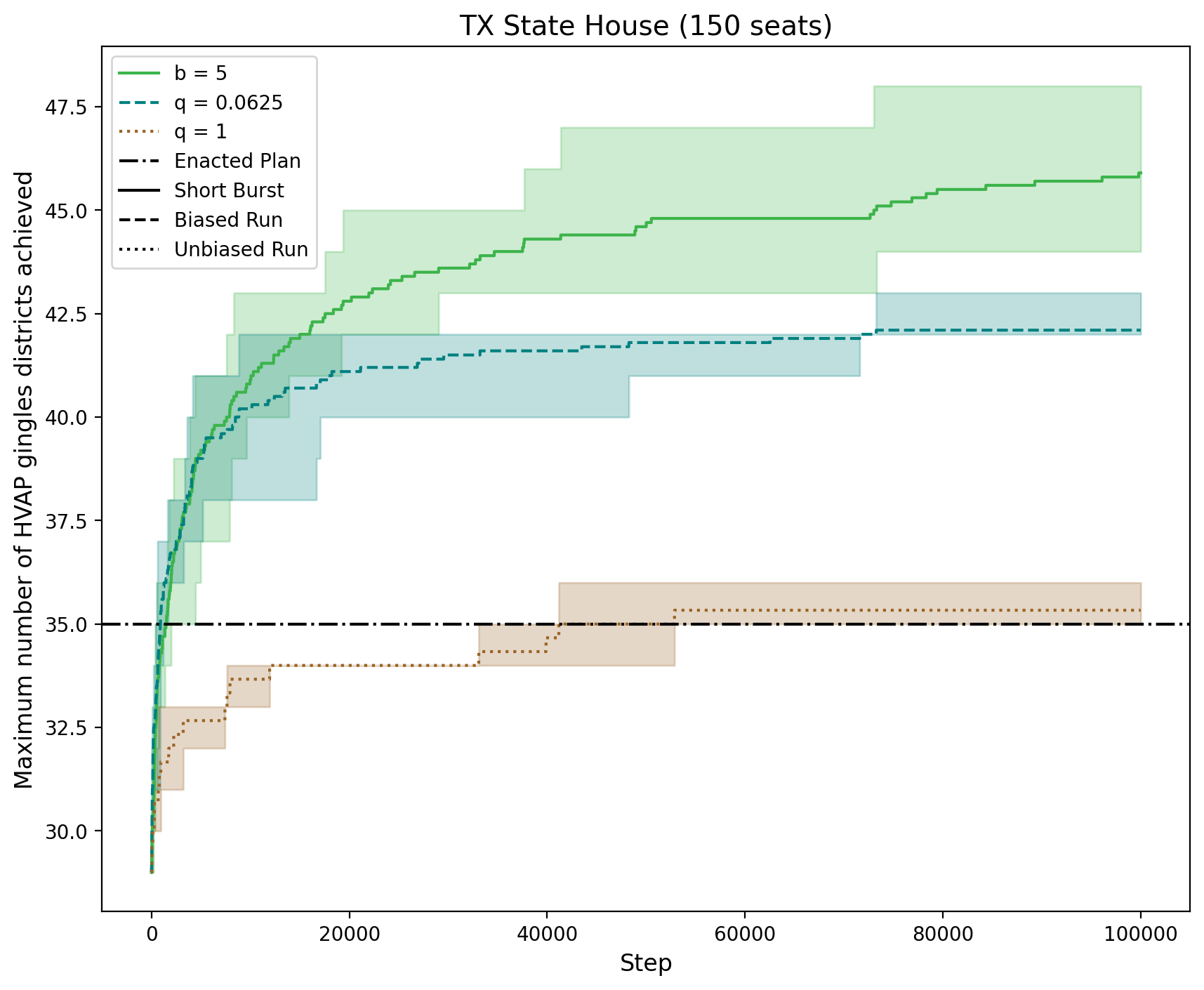}
            \caption{Maximums observed over time}
        \end{subfigure}
        \caption{HVAP}
    \end{subfigure}
    \begin{subfigure}{0.9\textwidth}
        \begin{subfigure}{0.44\textwidth}
            \centering
            \includegraphics[width=\textwidth]{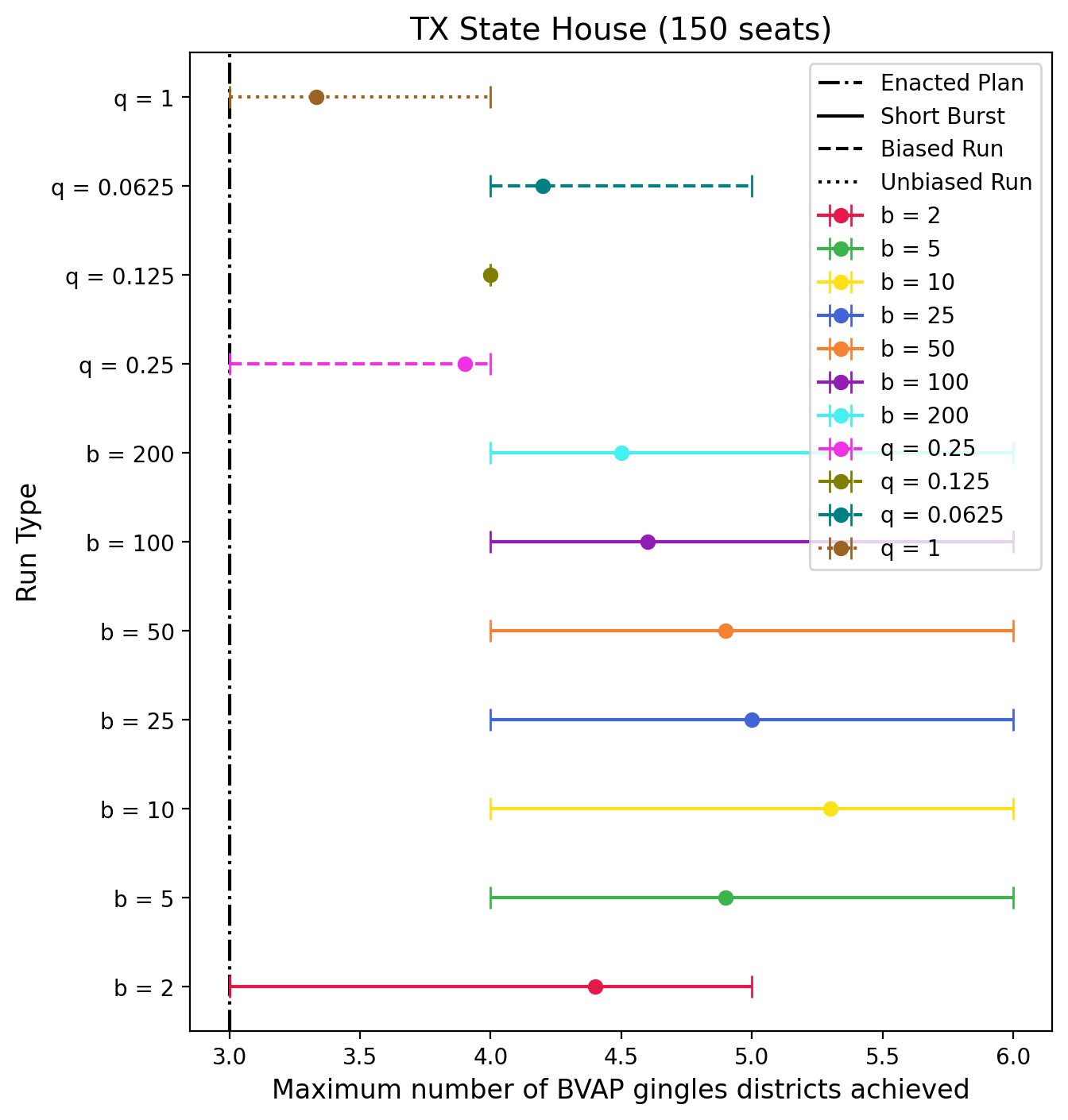}
            \caption{Maximum Gingles districts observed}
        \end{subfigure}
        \begin{subfigure}{0.55\textwidth}
            \centering
            \includegraphics[width=\textwidth]{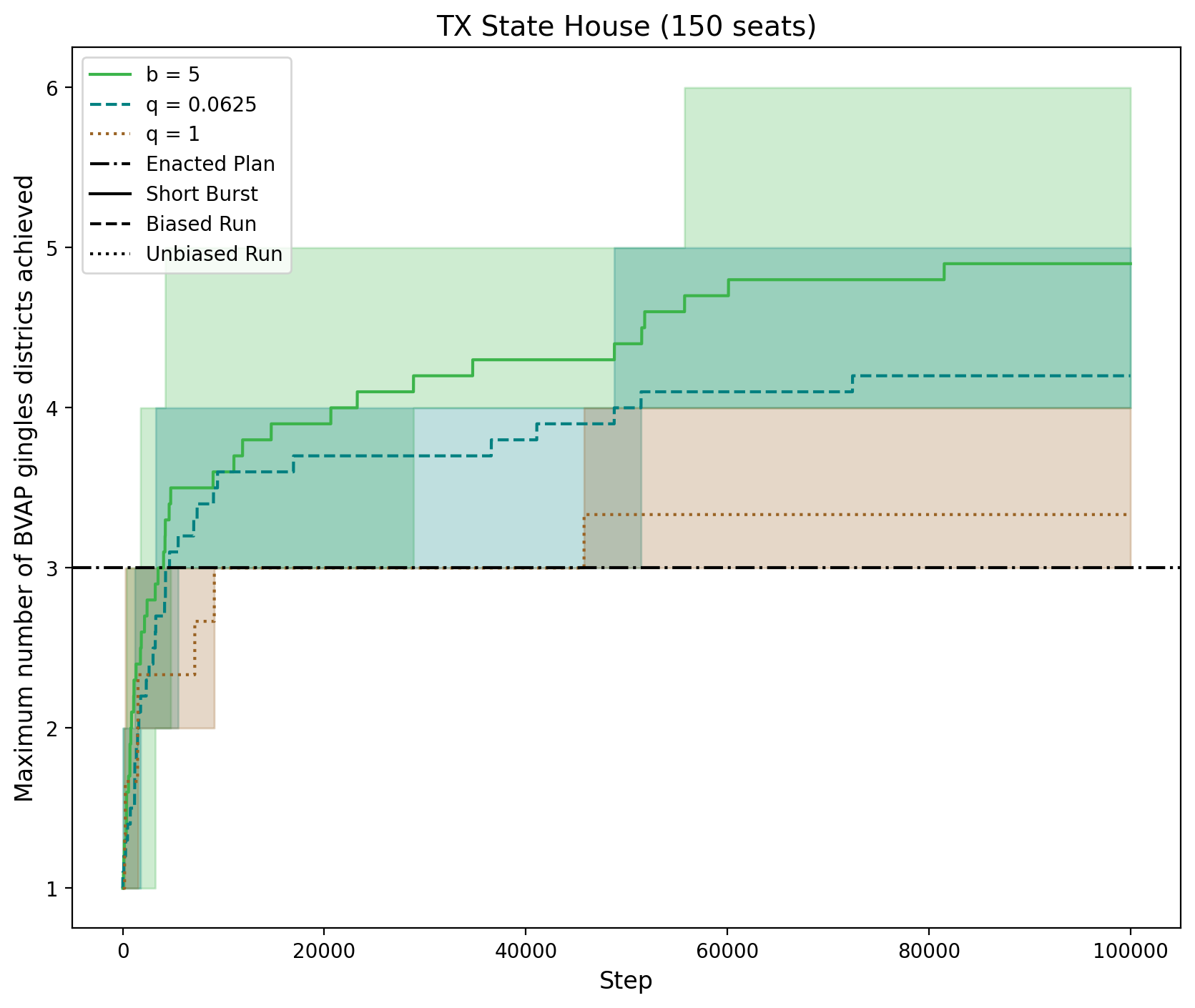}
            \caption{Maximums observed over time}
        \end{subfigure}
        \caption{BVAP}
    \end{subfigure}
    \begin{subfigure}{0.9\textwidth}
        \begin{subfigure}{0.44\textwidth}
            \centering
            \includegraphics[width=\textwidth]{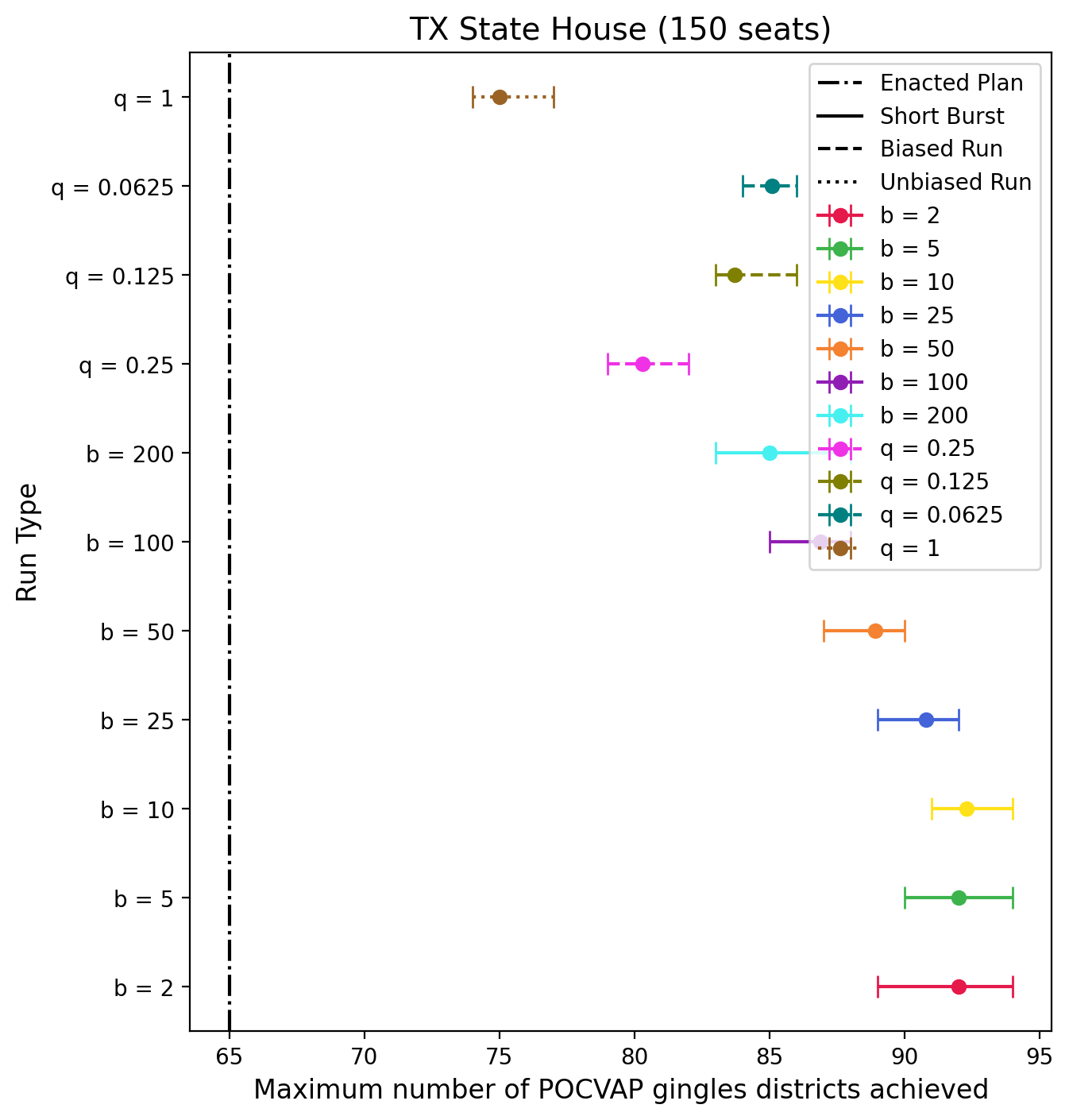}
            \caption{Maximum Gingles districts observed}
        \end{subfigure}
        \begin{subfigure}{0.55\textwidth}
            \centering
            \includegraphics[width=\textwidth]{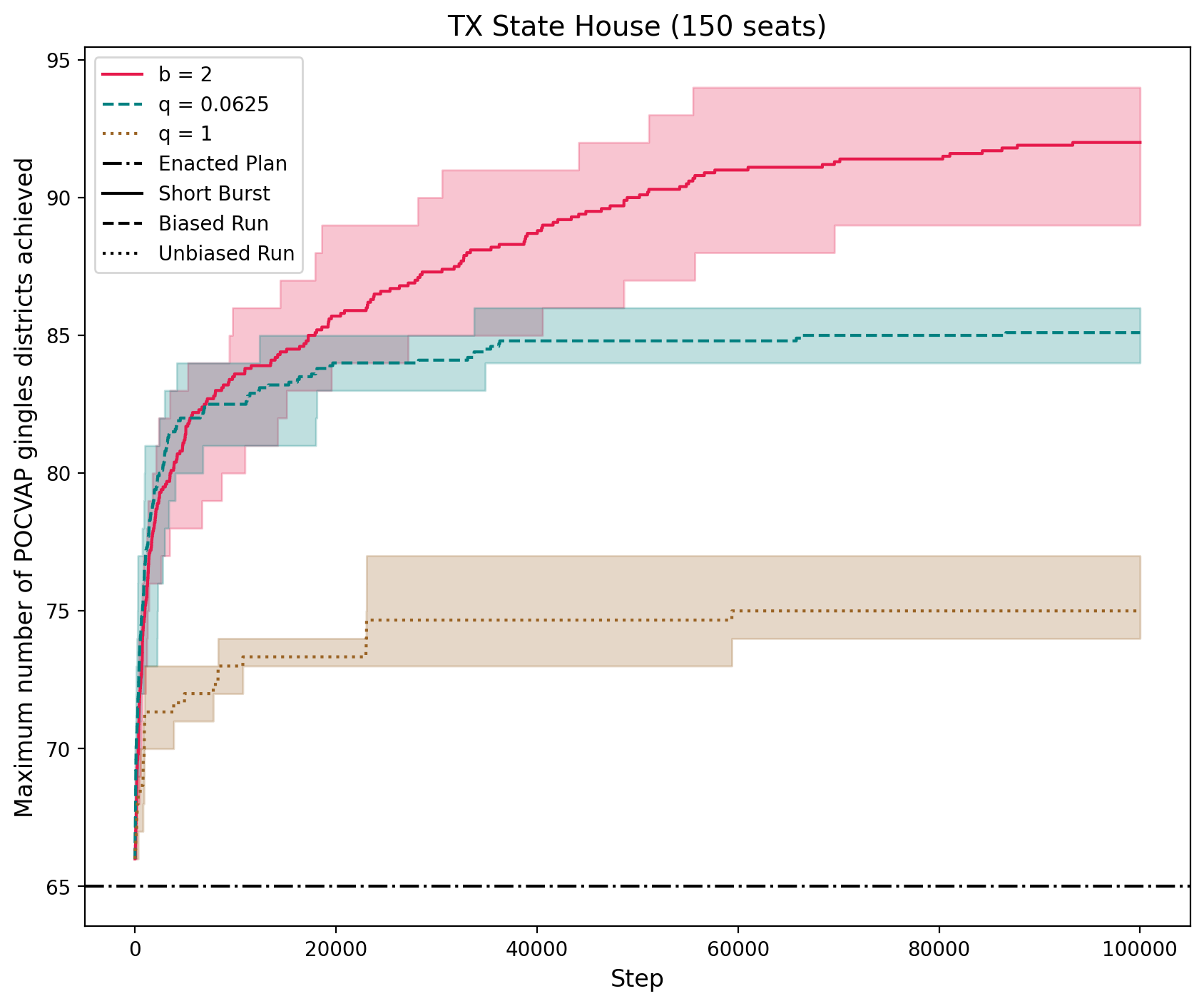}
            \caption{Maximums observed over time}
        \end{subfigure}
        \caption{POCVAP}
    \end{subfigure}

    \caption{Texas | Maximum numbers of majority-minority districts observed for short bursts, biased runs, and an unbiased run ($q$ = 1) for different minority populations. The short burst and biased 100,000 step runs were performed 10 times, and the unbiased run (q=1) was performed 3 times.  The left figures (a, d, g) show the range of maxes observed for each run type. The dot is the mean across the ten trials and the bars indicate the min/max range.  The right figures (b, e, h) show maximum number of majority-minority districts achieved at each step in the chain for the best preforming short burst length ($b$ = 5 or 2), biased run ($q$ = 0.0625), and the unbiased run ($q$ = 1).  The line indicates the mean, and the colored band the min/max range, across the trials.  The short bursts runs outperform the biased runs.
    }
    \label{fig:tx-maxes-all}
\end{figure}

For additional experiments in Virginia and New Mexico, see the appendix. In all cases, except when the population of interest was too small (that is, when the range of possible maximum values is too small), short bursts could be seen to outperform other methods.

\section{Explanatory Models} \label{sec:1D}

We have shown in Section \ref{section:empirical} that short bursts outperform biased and unbiased random walks for maximizing the number of Gingles districts in Louisiana.
In this section, we turn away from real-world data to simpler models to attempt to understand why this happens. We examine the effect of short bursts on the process of finding \oes with large scores for simple random walks in one dimension and on other uncomplicated state spaces. 

We begin by studying the simple random walk on $\mathbb{Z}$, and prove that short bursts and biased random walks are equivalent in expectation. We then examine a more realistic one-dimensional random walk where transition probabilities are not uniform, but depend on the current position: the farther from the origin a random walk is, the less likely it is to move away from the origin. Specifically, we consider a random walk on $\mathbb{Z}$ whose stationary distribution is approximately normal. In this case, short bursts do indeed outperform biased random walks. However, the effectiveness of short bursts strictly decreases as burst length increases, which is not what was observed in practice.  
Finally, we move beyond one dimension and consider random walks on a graph with a bottleneck. In this case, many burst lengths perform similarly, with some moderate-length bursts even outperforming shorter bursts.  

Altogether, these results suggest short bursts of a variety of lengths are successful for Louisiana for two reasons. First, short bursts are successful because the distribution of majority-minority districts is unimodal and approximately normal (Figure~\ref{fig:unbiased_results}, top).  Second, rather than seeing the effectiveness of bursts strictly decrease with bursts length, many short bursts lengths perform similarly well because there are small bottlenecks in the state space that the unbiased bursts must overcome -- provided the burst length is large enough to overcome these bottlenecks, we see similar results for different burst lengths.

\subsection{Short bursts for a Simple Random Walk on $\mathbb{Z}$}

In this section, we consider a simple random walk on $\mathbb{Z}$, where the score of a state is its label. This can be viewed as a one-dimensional projection of a larger Markov chains onto the (integer) score values of its states. Considering a simple random walk is equivalent to assuming that in the larger Markov chain, the scores are equally likely to increase or decrease in any given step. This is of course an over-simplification, but enables rigorous mathematical study of short bursts and a comparison to biased random walks. 
More realistic one-dimensional models are considered in the next subsection. 

\subsubsection{Rigorous Analysis of One Short Burst} \label{sec:one-unbiased-burst}

In this section, we rigorously analyze one short burst of length $b$ of a simple random walk. When we do not specify the starting point, we assume that the random walk starts at $X_0 = 0$. We are able to obtain a precise formulation for the probability distribution over the maximum achieved ($m$) in a random walk with burst of length $b$. We prove a recursive formula to describe the distribution of $m$ and show it has a closed form involving binomial coefficients. We then use that result to calculate the expected maximum value achieved by a short burst of length $b$. 

We briefly remark that $m$ is the maximum of a burst of length $b$ precisely when, in a simple random walk, the hitting time for $m$ is less than or equal to $b$ and the hitting time for $m+1$ is greater than $b$.  However, rather that working with hitting times, which would require careful consideration of conditional probabilities, we calculate the distribution of $m$ directly by setting up a recursion. 

We consider the maximum position $m$ achieved in a burst of length $b$ starting at 0. There are $2^b$ possible walks of length $b$ and we define $f(m,b)$ to be the number walks of length $b$ starting at $0$ that have $m$ as their maximum value, where $0 \leq m \leq b$. 

\begin{prop}
\label{recursive_prop}
For $m, b \in \mathbb{Z}$ and $b \geq 1$, the number $f(m, b)$ of bursts of length $b$ whose maximum  value achieved is $m$ satisfies:
$$
f(m, b) = 
\begin{cases} 
f(m-1, b-1) + f(m+1, b-1) &\mbox{if }  0 < m < b-1 \\
f(m, b-1) + f(m+1, b-1) &\mbox{if }  m = 0 \text{ and } b > 1 \\
1 &\mbox{if }  m = b-1 \text{ or } m = b \\
0 &\mbox{if }  m < 0 \text{ or } m > b
\end{cases}
$$
\end{prop}
\begin{proof}
We prove this by induction on $b$, the burst length.

\underline{Base Case ($b = 1$):}
The possible bursts of length 1 are one rightward step (from $0$ to $1$) or one leftward step (from $0$ to $-1$) with maximums of 1 and 0, respectively. Using our formula, we find that $f(0, 1) = 1$ and $f(1, 1) = 1$. This is the desired result.

\underline{Inductive Step:} For $b \geq 2$, we assume $f(m, b-1)$ satisfies the recursive formula for all values of $m$. We will show that $f(m, b)$ satisfies the recursive formula for all values of $m$. We do cases for the different values of $m$. 

\underline{Case 1 ($m = 0$):}
Since $m = 0$ and $b > 1$, we want to show that $f(0, b) = f(0, b-1) + f(1, b-1)$ from the recursive formula. We will count the number of the $2^b$ possible paths that result in a maximum value of 0. If the maximum value of the path is 0, then its first step must be leftward. We now consider the possibilities for the remaining $b-1$ steps which are exactly the $2^{b-1}$ possible paths of a random walk of length $b-1$. If we add a leftward step to the beginning of all such paths whose maximum value is 0 or 1, we find a set of paths of length $b$ whose maximum value is 0. There are no other paths of length $b$ whose maximum value is 0 because paths of length $b - 1$ with a maximum value of 2 or higher result in paths of length $b$ with a maximum value of 1 or higher. This verifies our result that $f(0, b) = f(0, b-1) + f(1, b-1)$.

\underline{Case 2 ($0 < m < b - 1$):}
Since $0 < m < b - 1$, we want to show that $f(m, b) = f(m-1, b - 1) + f(m+1, b-1)$. We will count the number of length $b$ possible paths that result in a maximum value of $m$. We create simple random walks of length $b$ by taking all the paths of length $b-1$ and adding a rightward or leftward step to the beginning. This will shift the maximum position reached up or down one integer from the maximum of the original path. If we take all the paths of length $b-1$ with a maximum value of $m+1$ and add a leftward step to the beginning of the path, we shift its maximum value to $m$. Likewise, if we take the paths with maximum $m-1$ and add a rightward step to the beginning, we shift the maximum value of the new path to $m$. Thus, the number of length $b$ paths with maximum $m$ is exactly equal to the number of length $b-1$ paths with maximum $m-1$ plus the number of paths with maximum $m+1$. Thus, $f(m, b) = f(m-1, b-1) + f(m+1, b-1)$.

\underline{Case 3 ($m = b-1$):}
Since $m = b-1$, we want to show that $f(b-1, b) = 1$ from the recursive formula. The only path of length $b$ whose maximum value is $b-1$ is a path with $b-1$ rightward steps followed by a leftward step. So, $f(b-1, b) = 1$.

\underline{Case 4 ($m = b$):}
Since $m = b$, we want to show that $f(b, b) = 1$ from the recursive formula. The only path of length $b$ whose maximum value is $b$ is a path with $b$ rightward steps. So, $f(b, b) = 1$.

\underline{Case 5 ($m < 0$ or $m > b$):}
We cannot have any walks where $m < 0$ since all walks begin at 0 which establishes a lower bound on $m$. We also cannot have any walks of length $b$ where $m > b$ as the most rightward path (all rightward steps) results in $m = b$.
\end{proof}

We now use this recursive formula to show a surprising connection between $f(m, b)$ and the binomial coefficients. Although $f(m,b)$ is not precisely a binomial distribution, it is a reshuffling of that distribution.
Consider the table below of $b$-step random walks and the frequency distribution of the maximums achieved:

\begin{center}
\begin{tabular}{ |c|c|c|c|c|c|c| } 
 \hline
 \textbf{b}/\textbf{m} & \textbf{0} & \textbf{1} & \textbf{2} & \textbf{3} & \textbf{4} & \textbf{5}  \\ 
 \hline
 \textbf{1} & 1 & 1 & & & & \\ 
 \hline
 \textbf{2} & 2 & 1 & 1 & & & \\ 
 \hline
 \textbf{3} & 3 & 3 & 1 & 1 & & \\ 
 \hline
 \textbf{4} & 6 & 4 & 4 & 1 & 1 & \\ 
 \hline
 \textbf{5} & 10 & 10 & 5 & 5 & 1 & 1 \\ 
 \hline
\end{tabular}
\end{center}
These are exactly the binomial coefficients; however, they are in descending order instead of the common symmetric order found, for instance, in Pascal's triangle. We will prove this relationship rigorously in the following proposition. Note that while there are simpler ways of expressing the following, we use the formula given to denote this alternating order of binomial coefficients.  
\begin{prop} \label{prop:closedform}
The distribution of $f(m, b)$ can be written as a permutation of the binomial coefficients as follows:
\[
f(m, b) = 
\begin{cases}
      {b \choose \frac{b}{2} + (-1)^m \lceil \frac{m}{2} \rceil } & \text{if $b$ is even} \\
      {b \choose \frac{b-1}{2} + (-1)^{m-1} \lceil \frac{m}{2} \rceil } & \text{if $b$ is odd}
\end{cases}
\]
\end{prop}
\begin{proof}
We will prove this for even $b$ and a similar methodology applies for odd $b$. We use induction on even values of $b$ to prove the result. 

\underline{Base Case:} Consider $b = 2$. We calculate 
\begin{align*}
    f(0, 2) = {2 \choose \frac{2}{2} + (-1)^0 \lceil \frac{0}{2} \rceil } = 2 \\
    f(1, 2) = {2 \choose \frac{2}{2} + (-1)^1 \lceil \frac{1}{2} \rceil } = 1 \\
    f(2, 2) = {2 \choose \frac{2}{2} + (-1)^2 \lceil \frac{2}{2} \rceil } = 1
\end{align*}
which we can verify in our table above.

\underline{Inductive Step:} 
Assume that for a fixed, arbitrary, even $b \geq 2$, $f(m, b) = {b \choose \frac{b}{2} + (-1)^m \lceil \frac{m}{2} \rceil }$ for $m$ satisfying $0 \leq m \leq b$. We will prove that $f(m, b+2) = {b + 2 \choose \frac{b+2}{2} + (-1)^m \lceil \frac{m}{2} \rceil }$ for $m$ satisfying $0 \leq m \leq b+2$ using the inductive assumption and Proposition \ref{recursive_prop}. We prove this in cases; when $m$ is close to 0 or close to $b$, care needs to be taken when applying Proposition \ref{recursive_prop}.

\underline{Case 1 ($m = 0$):}
\begin{align*}
    f(0, b + 2) &= f(0, b + 1) + f(1, b + 1) \\
    &= f(0, b) + f(1, b) + f(0, b) + f(2, b) \\
    &= 2f(0, b) + f(1, b) + f(2, b) \\
    &= 2{b \choose \frac{b}{2}} + {b \choose \frac{b}{2} - 1} + {b \choose \frac{b}{2} + 1} 
    = {b+1 \choose \frac{b}{2}} + {b+1 \choose \frac{b}{2} + 1} 
    = {b+2 \choose \frac{b+2}{2}}
\end{align*}
which is the desired result.

\underline{Case 2 ($m = 1$):}
\begin{align*}
    f(1, b + 2) &= f(0, b + 1) + f(2, b + 1) \\
    &= f(0, b) + f(1, b) + f(1, b) + f(3, b) \\
    &= f(0, b) + 2f(1, b) + f(3, b) \\
    &= {b \choose \frac{b}{2}} + 2{b \choose \frac{b}{2} - 1} + {b \choose \frac{b}{2} - 2} = {b+1 \choose \frac{b}{2}} + {b+1 \choose \frac{b}{2} - 1} = {b+2 \choose \frac{b}{2}}
\end{align*}
which is the desired result.

\underline{Case 3 ($2 \leq m \leq b - 2$):}
\begin{align*}
    f(m, b + 2) &= f(m - 1, b + 1) + f(m + 1, b + 1) \\
    &= f(m - 2, b) + f(m, b) + f(m, b) + f(m + 2, b) \\
    &= 2f(m, b) + f(m - 2, b) + f(m + 2, b) \\
    &= 2{b \choose \frac{b}{2} + (-1)^m \lceil \frac{m}{2} \rceil } + {b \choose \frac{b}{2} + (-1)^{m-2} \lceil \frac{m-2}{2} \rceil } + {b \choose \frac{b}{2} + (-1)^{m+2} \lceil \frac{m+2}{2} \rceil } \\
    &= 2{b \choose \frac{b}{2} + (-1)^m \lceil \frac{m}{2} \rceil } + {b \choose \frac{b-2}{2} + (-1)^m \lceil \frac{m}{2} \rceil } + {b \choose \frac{b+2}{2} + (-1)^m \lceil \frac{m}{2} \rceil } \\
    &= {b+1 \choose \frac{b}{2} + (-1)^m \lceil \frac{m}{2} \rceil } + {b+1 \choose \frac{b + 2}{2} + (-1)^m \lceil \frac{m}{2} \rceil } \\
    &= {b + 2 \choose \frac{b+2}{2} + (-1)^m \lceil \frac{m}{2} \rceil }
\end{align*}
which is the desired result.

\underline{Case 4 ($m = b-1$ or $b$):} 
\begin{align*}
    f(m, b + 2) &= f(m - 1, b + 1) + f(m + 1, b + 1) \\
    &= f(m - 2, b) + f(m, b) + 1 \\
    &= {b \choose \frac{b}{2} + (-1)^{m-2} \lceil \frac{m-2}{2} \rceil} + 2 = b + 2 = {b+2 \choose \frac{b+2}{2} + (-1)^m \lceil \frac{m}{2} \rceil}
\end{align*}
which is the desired result.

\underline{Case 5 ($m = b+1$ or $b+2$):}
\begin{align*}
    f(b + 1, b + 2) &= {b + 2 \choose \frac{b + 2}{2} + (-1)^{b+1} \lceil \frac{b+1}{2} \rceil } = {b + 2 \choose \frac{b + 2}{2} - \frac{b+2}{2}} = {b + 2 \choose 0} = 1 \\
    f(b + 2, b + 2) &= {b + 2 \choose \frac{b + 2}{2} + (-1)^{b + 2} \lceil \frac{b + 2}{2} \rceil } = {b + 2 \choose \frac{b+2}{2} + \frac{b+2}{2}} = {b + 2 \choose b + 2} = 1
\end{align*}
which matches our results from Proposition \ref{recursive_prop}.

The base case for the odd formula ($b = 1$) gives $f(0, 1) = {1 \choose 0} = 1$ and $f(1, 1) = {1 \choose 1} = 1$ which we can verify in our table above. A nearly identical approach using induction on the odd values of $b$ verifies the proposition in the odd case. Thus, we have proved that this permutation of binomial coefficients matches $f(m, b)$ exactly for all $b$ and all $0 \leq m \leq b$. 
\end{proof}

We now define a random variable for the maximum achieved in a burst of length $b$, and use the closed form formula of the previous proposition to calculate its expectation. 

\begin{definition} For $b \geq 1$, let $X_b$ be a random variable denoting the maximum achieved in a single short burst of length $b$ starting from 0. The distribution of $X_b$ is: 

\begin{align*}
   \p(X_b = m)  &=
    \frac{f(m, b)}{2^b} \\
    &= \frac{1}{2^b}
    \begin{cases}
      {b \choose \frac{b}{2} + (-1)^m \lceil \frac{m}{2} \rceil } & \text{if}\ b\ \text{is even} \\
      {b \choose \frac{b-1}{2} + (-1)^{m-1} \lceil \frac{m}{2} \rceil } & \text{if}\ b\ \text{is odd}
    \end{cases}
  \end{align*}
\end{definition}

\begin{prop}
\label{exp_max_prop}
For any $b \geq 1$, the expected value of the random variable $X_b$ is 
\begin{equation}
    \mathbb{E}[X_b] =  \frac{1}{2^b}
    \begin{cases}
    (b+\frac{1}{2}) {b \choose \frac{b}{2}} - 2^{b-1} & \text{if}\ b\ \text{is even} \\
    (b+1)  {b \choose \left\lfloor \frac{b}{2} \right\rfloor } - 2^{b-1}  & \text{if}\ b\ \text{is odd} \\
    \end{cases}
\end{equation}
\end{prop}
\begin{proof} We use the definition of expectation to see that 
 \begin{align*}
     \mathbb{E}[X_b] = \sum_{m=0}^{b} m \cdot \p(X_b = m) =\sum_{m=0}^{b} m \frac{f(m, b)}{2^b}
 \end{align*}
We consider the term $\sum_{m=0}^{b} m f(m, b) $.
Suppose $b$ is even; a similar methodology can be applied for odd $b$, but we omit the details. 
Using Proposition~\ref{prop:closedform}, the identity ${n \choose k} = {n \choose n - k}$, and the parity of $b$, we find that
\begin{align*}
     \sum_{m=0}^{b} m f(m, b) &= \sum_{m=0}^{b} m {b \choose \frac{b}{2} + (-1)^m \lceil \frac{m}{2} \rceil} =
    \sum_{m=0}^{b} m {b \choose \frac{b}{2} - \lceil \frac{m}{2} \rceil}
\end{align*}
Since our form includes $\lceil \frac{m}{2} \rceil$, we can group pairs of identical terms and change the index of summation as follows:
\begin{align*}
    \sum_{m=0}^{b} m {b \choose \frac{b}{2} - \lceil \frac{m}{2} \rceil} &= \sum_{i=1}^{\frac{b}{2}}\left( (2i-1) {b \choose \frac{b}{2} - i} + 2i {b \choose \frac{b}{2} - i}  \right)\\ &=
    \sum_{i=1}^{\frac{b}{2}} (4i - 1) {b \choose \frac{b}{2} - i}
\end{align*} 

Let $\ell = \frac{b}{2} - i$ and observe the following:

\begin{align*}
    \sum_{i=1}^{\frac{b}{2}} (4i - 1) {b \choose \frac{b}{2} - i } &=
    \sum_{\ell = 0}^{\frac{b}{2} - 1} \left(4\left(\frac{b}{2} - \ell\right) - 1\right) {b \choose \ell} \\
    &= (2b - 1) \sum_{\ell = 0}^{\frac{b}{2} - 1} {b \choose \ell}
     - 4\sum_{\ell = 0}^{\frac{b}{2} - 1} \ell {b \choose \ell}
\end{align*}
We then employ the common binomial identities ${n \choose k} = {n \choose n-k}$, $\sum_{k = 0}^{n} {n \choose k} = 2^n$, and $k {n \choose k} = n {n - 1 \choose k - 1}$ to find:
\begin{align*}
    (2b - 1) \sum_{\ell = 0}^{\frac{b}{2} - 1} {b \choose \ell}
     - 4\sum_{\ell = 0}^{\frac{b}{2} - 1} \ell {b \choose \ell} &= (2b-1) \left(\frac{2^b - {b \choose \frac{b}{2}}}{2}\right) 
     - 4\sum_{\ell = 0}^{\frac{b}{2} - 2} b {b-1 \choose \ell} \\
     &=(2b-1) \left(\frac{2^b - {b \choose \frac{b}{2}}}{2}\right)
     - 4b \left(\frac{2^{b-1} - {b-1 \choose \frac{b}{2} - 1} - {b-1 \choose \frac{b}{2}}}{2}\right) \\
     &= b2^b - 2^{b-1} + \left(\frac{1}{2} - b\right){b \choose \frac{b}{2}} - 4b\left(2^{b-2} - \frac{1}{2}{b \choose \frac{b}{2}}\right) \\
     &= \left(b + \frac{1}{2}\right) {b \choose \frac{b}{2}} - 2^{b-1} 
\end{align*}
So, we have found that for even $b$,
\begin{align*}
     \sum_{m=0}^{b} m f(m, b) &= \left(b + \frac{1}{2}\right) {b \choose \frac{b}{2}} - 2^{b-1} 
\end{align*}
     That is, we have shown 
\begin{align*}
      \mathbb{E}[X_b] &= \frac{1}{2^b}\left(\left(b + \frac{1}{2}\right) {b \choose \frac{b}{2}} - 2^{b-1}\right)
\end{align*}
The proof for odd $b$ is nearly identical. 

\end{proof}
We note that for a simple random walk, it does not matter where the burst starts. The expected maximum displacement to the right is the same. We demonstrate this behavior in the following lemma which we will use in the next subsection, which looks at multiple short bursts.

While we continue to use random variable $X_b$ to denote the maximum value achieved after a short bursts of length $b$, below we use random variable $X_0$ to denote the starting position of the burst, which can be thought of as the maximum achieved after $0$ steps. 

\begin{lemma}
Let $x_0$ be the initial position of the short burst. Then
\begin{equation}
\mathbb{E}[X_b|X_0 = x_0] = \mathbb{E}[X_b|X_0 = 0] + x_0
\end{equation}
\end{lemma}
\begin{proof}
The possible values for $X_b$ range from $x_0$ to $x_0+b$, so 
\begin{align*}
\mathbb{E}[X_b|X_0 = x_0] &= \sum_{m=x_0}^{x_0+b} m \cdot \p(X_b = m|X_0 = x_0)
\end{align*}
Because the transition probabilities in an unbiased short burst are invariant to the position, we can shift all probabilities to start at $X_0 = 0$. In other words, we know
\begin{align*}
\p(X_b = x_b|X_0 = x_0) = \p(X_b = x_b - x_0|X_0 = 0)
\end{align*}
This fact leads us to the following
\begin{align*}
\sum_{m=x_0}^{x_0+b} m \cdot \p(X_b = m|X_0 = x_0) &= \sum_{m=x_0}^{x_0+b} m \cdot \p(X_b = m - x_0|X_0 = 0)
\end{align*}
If we let $m' = m - x_0$, we find
\begin{align*}
\sum_{m=x_0}^{x_0+b} m \cdot \p(X_b = m - x_0|X_0 = 0) &= \sum_{m'=0}^{b} (m' + x_0) \cdot \p(X_b = m'|X_0 = 0) \\
&= \sum_{m'=0}^{b} m' \cdot \p(X_b = m'|X_0 = 0) + x_0 \sum_{m'=0}^{b} \p(X_b = m'|X_0 = 0) \\
&= \mathbb{E}[X_b|X_0 = 0] + x_0.
\end{align*}
\end{proof}

This lemma shows that short bursts according to a simple random walk are invariant to position.  Thus, the expected maximum displacement for a short burst of length $b$ is also invariant to starting position. A series of short bursts in this regime can then be described as a set of independent single short bursts whose starting point is the maximum value of the previous burst. This useful independence makes one-dimensional simple random walks a good first candidate for understanding the behavior of short bursts, but is unlikely to occur in more complicated systems. 

\subsubsection{Analysis of multiple short bursts and effects of burst length}

In this section, we compare a series of short bursts with biased and unbiased simple random walks. We see that in this one-dimensional model, short bursts are essentially equivalent to biased random walks. All of the code used for these experiments is available at \url{https://github.com/vrdi/shortbursts-gingles/tree/master/toy_models}.

\begin{figure}
    \centering
    \includegraphics[width=0.8\textwidth]{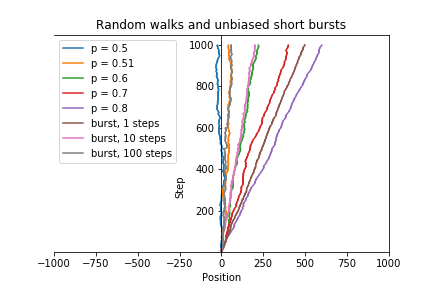}
    \caption{Comparing unbiased short burst runs of different lengths with unbiased and biased simple random walks. All of the walks are run for a total of 1000 steps shown on the y-axis. The position of these walks is shown on the x-axis.}
    \label{compare_bursts_walks}
\end{figure}

Figure \ref{compare_bursts_walks} shows walks on the integer line using short bursts, biased, and unbiased random walks. We plot one unbiased simple random walk ($p = 0.5$) and a number of biased random walks where $p$ is the probability of stepping rightward. We run each of these random walks for 1000 steps. We also plot short bursts where the product of the number of steps in the burst, $b$, and the total number of bursts, $k$, is constant at $n = 1000$. Consequently, each total short burst run is 1000 steps.

We observe that the more biased the random walk, the further it travels rightward. For short burst runs, as $b$ decreases, it travels further rightward. This behavior becomes clear when we consider the extreme run with $b = 1$ where the walk can only move to the right or stay the same. From Figure \ref{compare_bursts_walks}, it appears that biased random walks and unbiased short burst runs exhibit similar behavior. For example, the biased run with $p = 0.6$ overlaps significantly with a short burst with $b = 10$. 

We prove that for every short burst run, its expected position after a given number of steps is the same as the expected position of a biased random walk with appropriately chosen bias $p$. We note the same would be true for any random process whose expected position after $n$ steps is linear in $n$. 

\begin{prop}
\label{equivalence_biased_bursts}
Consider integers $n$ and $b$ where $b$ is a divisor of $n$. Let $D_p^{n}$ be the position of a biased random walk with bias $p$ after $n$ steps. Let $D_b^{n}$ be the position of an unbiased short burst run after $n/b$ bursts, each of length $b$. Suppose $p$ satisfies
\begin{equation}
p = \frac{\mathbb{E}[X_b]}{2b} + \frac{1}{2}
\end{equation} 
where $\mathbb{E}[X_b]$ is the expected maximum of one short burst of length b from Prop $\ref{exp_max_prop}$.

Then $\mathbb{E} [D_p^{n}] = \mathbb{E}[D_b^{n}]$.
\end{prop}

\begin{proof}
We will derive expressions for $\mathbb{E}[D_p^{n}]$ and $\mathbb{E}[D_b^{n}]$ in terms of $p$ and $b$. Let's first consider the biased random walk. Since each step in a biased random walk is independent, and the expected change in position after one step of a biased random walk is $(-1)(1-p) + (1)(p) = 2p -1$,

\begin{align*}
    \mathbb{E} [D_p^{n}] 
    = n \mathbb{E} [D_p^{1}] = n(2p - 1)
\end{align*}

Now, let's consider the unbiased short burst run. Since each of the $n/b$ short bursts are independent, we know

\begin{align*}
    \mathbb{E} [D_b^{n}] = \frac{n}{b} \mathbb{E} [D_b^{b}] 
\end{align*}

The expected position of an unbiased run after one short burst of length $b$ is identical to the expected maximum achieved over a short burst of length $b$. In other words, $\mathbb{E} [D_b^{b}] = \mathbb{E}[X_b]$. So, we know

\begin{align*}
    \mathbb{E} [D_b^{n}] = \frac{n}{b} \mathbb{E} [X_b] 
\end{align*}

Assuming the relationship in the proposition and putting this all together, we find

\begin{align*}
     \mathbb{E} [D_p^{n}] 
     &= n(2p - 1) \\
     &= n\left(2\left(\frac{\mathbb{E}[X_b]}{2b} + \frac{1}{2}\right) - 1\right) \\
     &= n\left(\frac{\mathbb{E}[X_b]}{b} + 1 - 1\right) \\
     &= \frac{n}{b} \mathbb{E} [X_b] \\
     &= \mathbb{E} [D_b^{n}]
\end{align*} \qed
\end{proof}

Proposition \ref{equivalence_biased_bursts} indicates that given a short burst run, there is a biased random walk with the same expected position after $n$ steps.
These are standard calculations to find the value $p$ such that a biased random walk has a desired expectation; for example, the same calculations are used to determine the correct value $p$ to approximate biased Brownian motion with a random walk, and the formula we obtain is similar.


We note that the shortest possible burst length ($b = 1$) is equivalent to $p = 0.75$, so in this one-dimensional setting, short burst runs cannot match biased runs with $p > 0.75$.

This indicates that simple random walks, while a logical first choice for exploration, do not help explain why short bursts outperform biased random walks as in Section~\ref{section:empirical}. A main feature of the explorations in Section~\ref{section:empirical} not captured by simple random walks is that transition probabilities may depend on score value. As we find districting plans with more majority-minority districts, redistricting plans with even more such districts are harder to find. A simple random walk, while mathematically appealing, does not capture this behavior. In the next section, we implement the short burst methodology over a more complex distribution.

\subsection{Normal Model}

In this section we move beyond a simple random walk on a line to a random walk on a line where transition probabilities depend on position: the farther away from the origin a random walk is, the less likely it is to move even farther from the origin. Specifically, we pick probabilities so that the stationary distribution of this random walk is approximately a normal distribution centered at 0. This is a better model of the transition probabilities between districts with various numbers of majority-minority districts in real state data. Figure \ref{fig:unbiased_results}(top) shows a histogram, based on an ensemble baseline in Lousiana, of the number of districts in which over 50\% of the voting age population is black. This histogram resembles a normal distribution, justifying our choice of a normal stationary distribution for the one-dimensional random walk. While of course such a projection onto one dimension hides other structural elements of the state space, a one-dimensional random walk with transition probabilities drawn to ensure the stationary distribution is normal is a better one-dimensional model than simple random walks. Unfortunately we cannot make rigorous comparisons between short bursts and biased random walks under this normal model, but we present extensive empirical investigations showing short bursts do indeed outperform biased random walks.  
\subsubsection{Random Walks according to the Normal Distribution}
\label{section:random_walks_normal_dist}
Consider a random walk $\mathcal{N}$ on the line that moves according to a normal distribution. Let $\pi$ be the stationary distribution and $P(i, j)$ represents the probability to move from state $i$ to state $j$. We will define the stationary distribution, $\pi$, using the normal probability density function as follows:

\begin{align*}
    \pi(i) = \int_{i - \frac{1}{2}}^{i + \frac{1}{2}} N(\mu, \sigma^2) dx
\end{align*}

In this case, we set $\mu$ = 0 to center our random walk at 0 and fit $\sigma$ to 1.44 from our distribution on Louisiana BVAP districts in Figure \ref{fig:unbiased_results}.

From the detailed balance conditions for a reversible random walk, we know $\pi$ and $P$ should satisfy  $\pi(i)P(i, i+1) = \pi(i+1)P(i+1, i)$. We use this to appropriately set the transition probabilities of this random walk $\mathcal{N}$, for simplicity assuming the probability of moving towards 0 from either direction is always 0.5. This leads us to the following transition probabilities for positive $i$:

\[ \begin{cases} 
      P(i, i+1) = \frac{\pi(i+1)}{2\pi(i)}\\
      P(i, i-1) = 0.5 \\
      P(i, i) = 1 - P(i, i+1) - P(i+1, i) \\
   \end{cases}
\]

Similarly for negative $i$, we find:

\[ \begin{cases} 
      P(i, i+1) = 0.5 \\
      P(i, i-1) = \frac{\pi(i-1)}{2\pi(i)} \\
      P(i, i) = 1 - P(i, i+1) - P(i+1, i) \\
   \end{cases}
\]

Finally, for $i = 0$, we find:

\[ \begin{cases} 
      P(0, 1) = \frac{\pi(1)}{2\pi(0)} \\
      P(0, -1) = \frac{\pi(-1)}{2\pi(0)} \\
      P(0, 0) = 1 - P(0, 1) - P(0, -1) \\
   \end{cases}
\]

\subsubsection{Analysis of One Short Burst for Different $X_0$}

For this random walk $\mathcal{N}$, we run simulations of one short burst on the line from different starting points of various lengths.  All of the code used for these experiments is available at \url{https://github.com/vrdi/shortbursts-gingles/tree/master/toy_models}. In Figure \ref{fig:expected_value_normal_diff_starting_points}, we compare these bursts starting at $X_0 = 0, 1$ and $2$. For each starting point, we numerically find the expected maximum value for a range of burst lengths. For each $b$ from 1 to 50, we run 1000 simulated random walks of length $b$ starting from $X_0$ and take the average maximum value achieved. Each of these curves is concave which implies that as the burst length increases, the expected maximum value is increasing less rapidly. When we compare the curves for different $X_0$, we see that as $b$ increases, the expected maximum values for different $X_0$ become closer. Additionally, the expected maximum value curves become flatter and less concave as $X_0$ increases, indicating there is less added benefit for larger short bursts. 
This is unsurprising: because the tails in a normal distribution drop off quickly, the probability of rightward steps decreases greatly as the random walk $\mathcal{N}$  goes further right. Once $\mathcal{N}$ moves far enough to the right, the history of the starting point becomes negligible and the power of the tails of the normal distribution take over.

\begin{figure}
\centering
    \includegraphics[width=0.8\textwidth, trim = 0 17 0 0, clip]{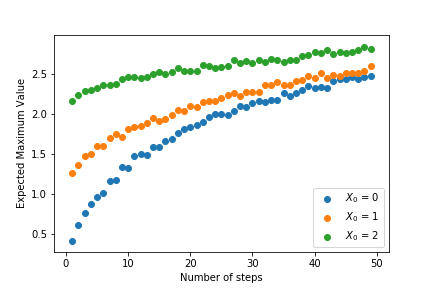}
    
    {\normalsize \fontfamily{cmss}\selectfont Number of steps in burst}
    
    \caption{Numerical estimation of the expected maximum value for a short burst whose length $b$ is the number of steps. For each data point, we use 1000 simulated short bursts of length $b$ (number of steps) starting from $X_0$ = 0, 1, or 2.} \label{fig:expected_value_normal_diff_starting_points}
\end{figure}

For random walk $\mathcal{N}$ (and in fact, in any one-dimensional model), it is better to take many short bursts with a smaller $b$ than a few short bursts with a larger $b$. This can be explained mathematically by coupling two short burst runs. For example, consider the first to be one short bursts of 50 steps of $\mathcal{N}$ and the second to be two short bursts for $\mathcal{N}$ each being 25 steps. Our coupling only requires that if the two runs are at the same position at the same time step, they make the same move in the next time step. For the first 25 steps, these two runs move in sync. Then, the second run will immediately jump to the right to the burst maximum. The second run will then always be to the right of or at the same position as the first run. The second run may be more likely to move left since it is farther from the origin and thus further into the tails of the distribution. However, the first run will never move right of the second because of the coupling. This shows that short bursts with smaller burst length $b$ will never be worse on average than short bursts with larger burst length $b$. In fact, in all observed cases, for the same number of total steps, longer burst lengths give strictly worse performance than shorter burst lengths. This is a feature of one-dimensional random walks; for an example of a graph for which many burst lengths perform similarly, see Section~\ref{sec:bottleneck}.

\subsubsection{Comparison of Short Bursts to Unbiased and Biased Random Walks}

In this section, we run simulations to compare random walk $\mathcal{N}$ with biased random walks and short bursts for the same process. The code used for these experiments is available at \url{https://github.com/vrdi/shortbursts-gingles/tree/master/toy_models}. For this biased random walk, we implement a rejection sampling algorithm just as done in Section~\ref{section:empirical}. First, we propose a move based on the probabilities of $\mathcal{N}$. If the move increases the score (position) or keeps it the same, we accept it. If the move would decrease the score, we accept it with some probability $q$ where $0 < q < 1$. If we don't accept it, we count the step as staying in the same position.

The parameter $q$ can be varied and represents the relative probability of moving left compared to the probability of moving right or staying constant.  There is a straightforward relationship between $q$ and the bias parameter, $p$, from biased simple random walks in the previous section. For biased simple random walks, $p$ is the probability of moving right and $1-p$ is the probability of moving left. Since $q$ is ratio between these two probabilities, we find $q = \frac{1-p}{p}$ or similarly, $p = \frac{1}{1+q}$. We use this relationship to equate biased simple random walks and biased random walks according to the normal distribution.

We ran simulations of a fixed total length for random walks and a variety of short burst runs for random walk $\mathcal{N}$. We run 1000 unbiased and biased random walks for 20000 steps and 1000 short burst runs each with a total length of 20000 steps (keeping the product of the burst length $b$ and the number of bursts $k$ constant at 20000). Figure \ref{fig:random_walks_normal} shows the average maximum position achieved for each of these over time for each of these. When plotting, we average over all 1000 trials and display a cone of 1 standard deviation around the mean. 

\begin{figure}
    \centering
        \includegraphics[width=0.95\textwidth, trim = 0 74 0 120, clip]{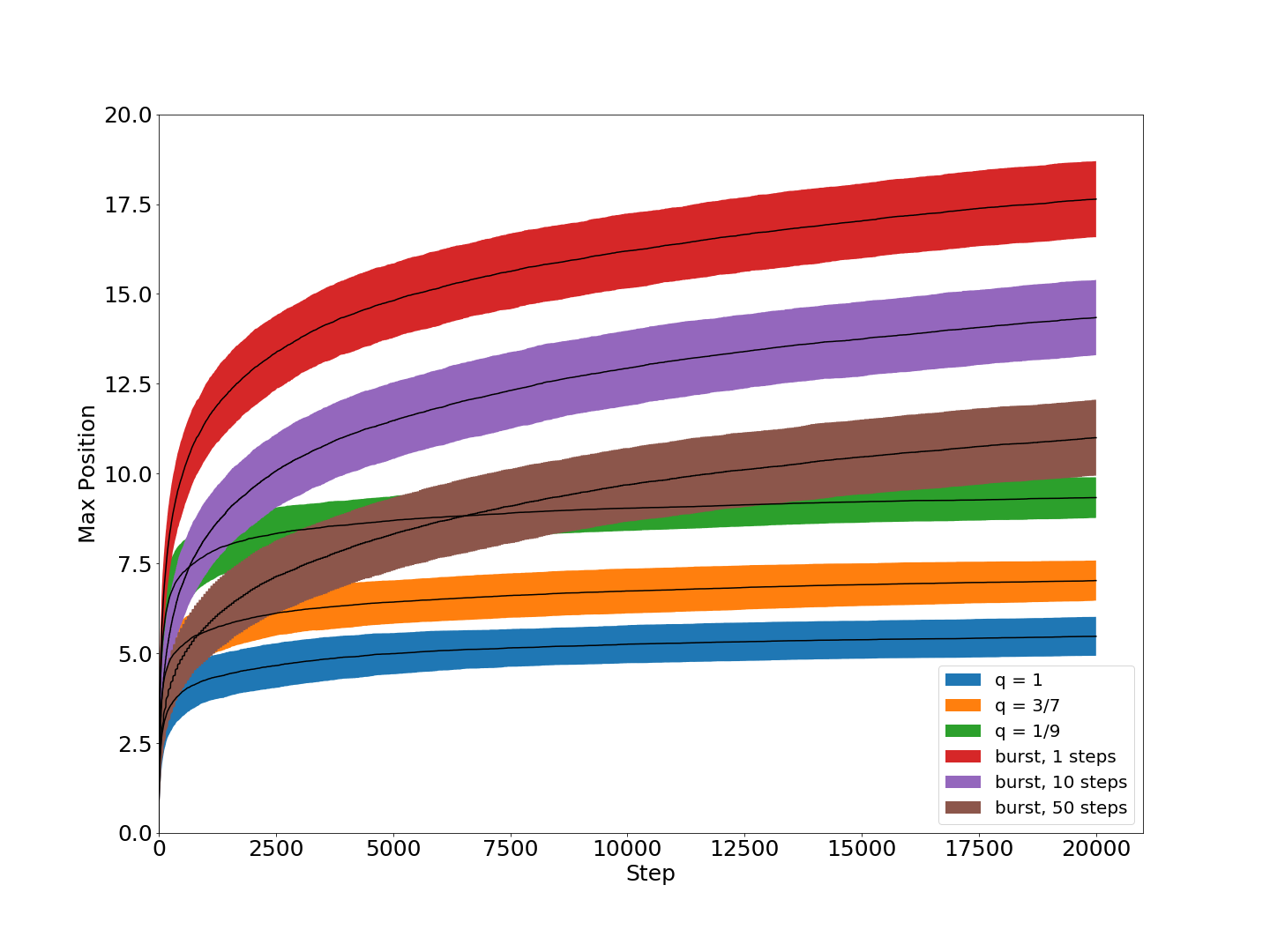}
            \caption{Comparing short bursts of various burst lengths to biased random walks (with probability $q$ of accepting a move that decreases the score function) and an unbiased random walk ($q = 1$) of Markov chain $\mathcal{N}$.  
         For each, the average maximum position achieved over 1000 independent trials is plotted; the shaded regions are cones of 1 standard deviation around the mean.  } \label{fig:random_walks_normal}

\end{figure}

We see that the short burst runs do much better than the unbiased and biased walks under the normal distribution. All of the short burst runs ($b = 1, 10, 50$) have maximum positions that are larger than any biased or unbiased random walks. Similarly to the simple case, short burst runs achieve higher maximums as the burst length decreases. However, even the short burst run with the longest burst length ($b = 50$) achieves a higher maximum position than the most biased random walk with $q = 1/9$ (or, in analogy to the previous section, $p = 0.9$). This is due to the rapidly diminishing tails of the normal probability distribution. After each short burst, the maximum value attained can never be less than the maximum of the previous short burst, whereas this is not the case for biased random walks. For short bursts, this has the effect of `pinning' the random walk far out in the tails of the normal distribution.  Thus, over time, the short bursts outperform biased random walks. This helps explain why short bursts outperform biased random walks for finding majority-minority districts, where the distribution of the number of such districts appears to be approximately normal. 

We do not believe there is anything special about a normal distribution, and that a similar effect would be found for any unimodal distribution with small tails. For instance, for a random walk on a line where probabilities of moving away from the origin decrease geometrically with the distance from the origin, a similar effect was observed; we found that after enough steps, short burst walks of varying burst lengths under a geometric distribution have maximum positions that are larger than biased or unbiased random walks.

\subsection{Beyond one dimension: Bottleneck Models} \label{sec:bottleneck}

In this section, we move beyond one-dimensional models.  Looking only at the one-dimension projection of a score function can hide other features of a state space that may be affecting how well various approaches perform.  

In the one dimensional models, we saw that shorter bursts always outperformed longer bursts, often significantly so. However, this is not what we observed in the experiments of Section~\ref{section:empirical}, where bursts of many different lengths had similarly strong performances.\footnote{Due to a slight inconsistency in notation, the bursts with $b = 2$ shown in the plots of Section~\ref{section:empirical} correspond to bursts consisting of two plans, a start plan and an end plan, which in this section have been called bursts of length $1$. These bursts are equivalent to a greedy walk which only accepts a new plan if it does not decrease the score function.} 
We believe this is specific to the line and present a scenario where many bursts of longer length perform similarly, and in fact even outperform the shortest bursts. 

In this section, we consider a graph with a bottleneck, where reaching the other side of the bottleneck -- where scores are higher -- requires first traversing a path with lower scores. On such a graph, biased walks perform poorly, as do short bursts whose burst length is shorter than or of comparable length to the length of the low-scoring path. Greedy random walks, where a new plan is only accepted if its score is the same or higher, also fail to cross this bottleneck. 

\begin{figure}
    \centering
     \includegraphics[width=0.5\textwidth]{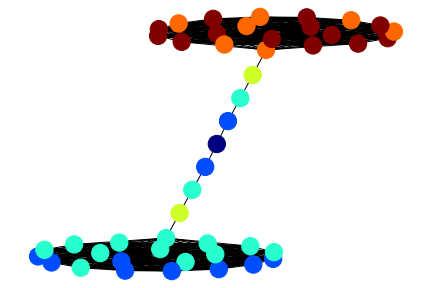}
     \caption{Bottleneck Graph: One copy of $K_{20}$ with scores randomly assigned to be $-1$ (light blue) or $-2$ (blue),  a second copy of $K_{20}$ with scores randomly assigned to be $1$ (orange) or $2$ (red), joined by a path of length 7 with scores of: 0, -1, -2, -3, -2, -1, 0.}
     \label{fig:bottleneck_graph}
\end{figure}

The bottleneck graph we consider consists of two 20-vertex complete graphs $K_{20}$ joined by a path $P_7$ of length $7$; see Fig.~\ref{fig:bottleneck_graph}
In one $K_{20}$, which we call the positive copy of $K_{20}$, each node has a score of 1 or 2. In the other $K_{20}$, which we call the negative $K_{20}$, nodes have scores of -1 or -2. The $P_7$ path has nodes with scores of $[0, -1, -2, -3, -2, -1, 0]$.  The first vertex in this path is connected to a vertex in the positive $K_{20}$ which has score 1, and the last vertex of this path is connected to a vertex in the negative $K_{20}$ with score $-1$.  Note all scores are integers and scores between adjacent vertices differ by at most one. 

 For this bottleneck graph and its associated scores, we simulated the behavior of unbiased random walks, biased random walks, short bursts, and a greedy random walk (equivalent to short bursts with $b=1$ or a biased random walk with $q=0$). We ran 100 simulations each for a fixed total length of 2400. 
 Figure \ref{fig:bottleneck_results} shows the maximum position achieved for each of these over time for each of these. Each line represents the average of 100 runs, and the shaded region is a 95\% confidence interval for that average.  The start node at the beginning of all simulations was a random node in the negative $K_{20}$ with a score of $-1$.

Figure \ref{fig:bottleneck_opts} shows that a short burst of length 20 significantly outperforms the other methods. By step 2400, all 100 runs have found the maximum possible score. Biased random walks have varying levels of success by time step 2400, where it seems that the higher the value of $q$, i.e the higher the chance of accepting a move that decreases the score, the better you are expected to perform. The greedy method using a burst length of 1, which only accepts a move to a score greater than or equal to the current score, never manages to get out of the negative $K_{20}$ and across the bottleneck.

Having established that bursts perform better than biased and unbiased random walks, Figure \ref{fig:bottleneck_bs} compares the different burst lengths with each other. Bursts of length 10, 20, 40, and 100 are able to find a state with the highest score of 2 in nearly all experiments, though medium-sized bursts of length 10, 20, and 40 seem to do so more quickly than bursts of length 100. When the burst length is shorter than the path $P_7$, the positive side of the graph is never reached. This example suggests that bottlenecks in the state space are one possible explanation for why we see short bursts of many different moderate lengths performing similarly in our case study for Louisiana rather than seeing strictly decreasing performance for longer burst lengths. In this example we also see bursts of shorter lengths (2 and 5) performing worse than shorter busts of moderate lengths.  This is something we did observe to a small extent in our experiments that was unexplained by our one-dimensional explorations, though the effect is certainly more dramatic here. One possible explanation is that there are bottlenecks of a variety of sizes within the state space, some of which short bursts of length 2 and 5 can still cross, but some of which are longer and are best traversed by short bursts of longer lengths. 

Similar effects were observed for sample graphs with more than one bottleneck. Short bursts with a burst length longer than any of the bottlenecks were almost always able to successfully cross all bottlenecks, while biased random walks, unbiased random walks, and short bursts with a burst length less that the bottleneck length were not. 

\begin{figure}
    \centering
    \begin{subfigure}{0.49\textwidth}
        \includegraphics[width=\textwidth]{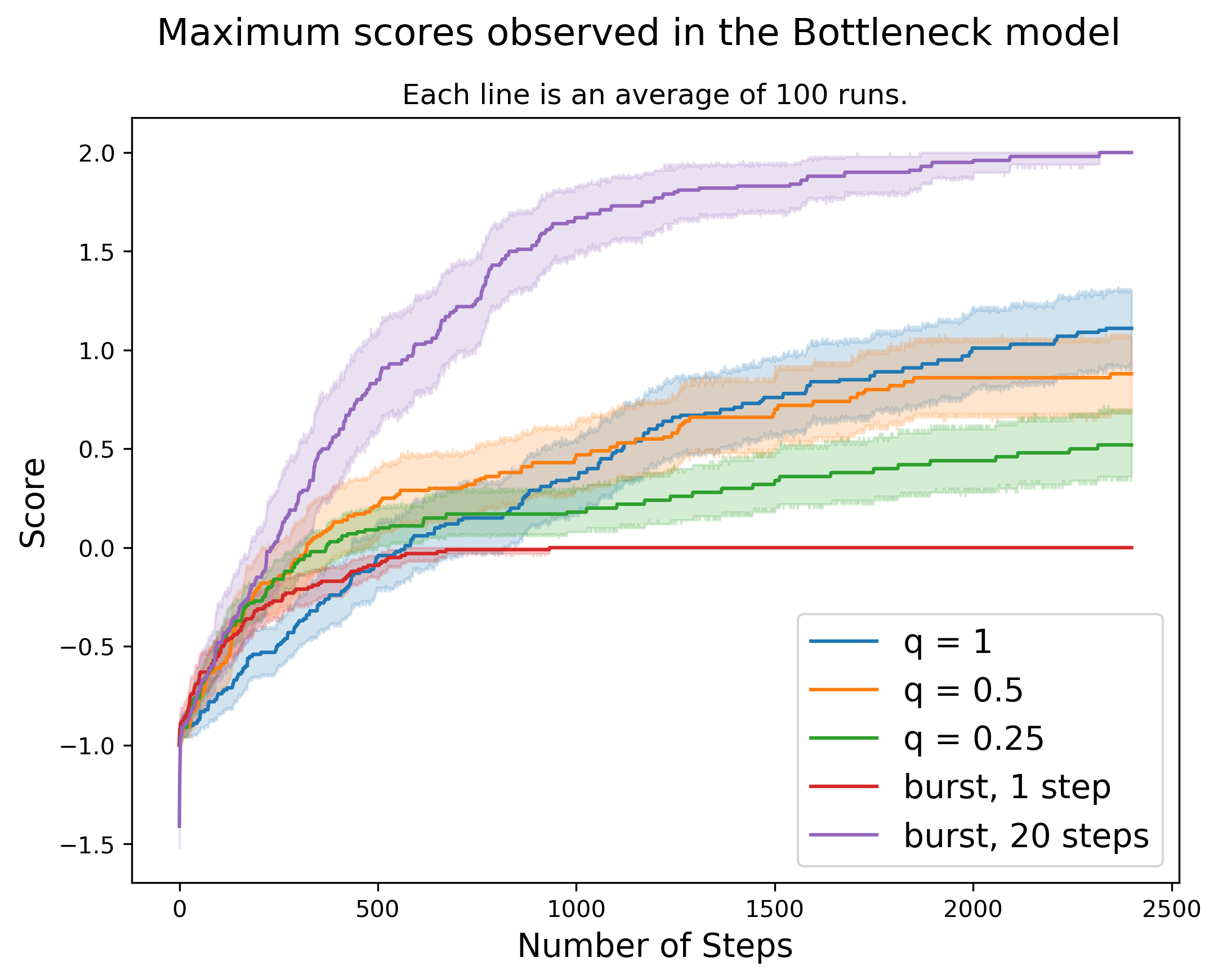}
        \caption{Short bursts vs. biased random~walks}
        \label{fig:bottleneck_opts}
   \end{subfigure}
   \centering
        \begin{subfigure}{0.49\textwidth}
        \includegraphics[width=\textwidth]{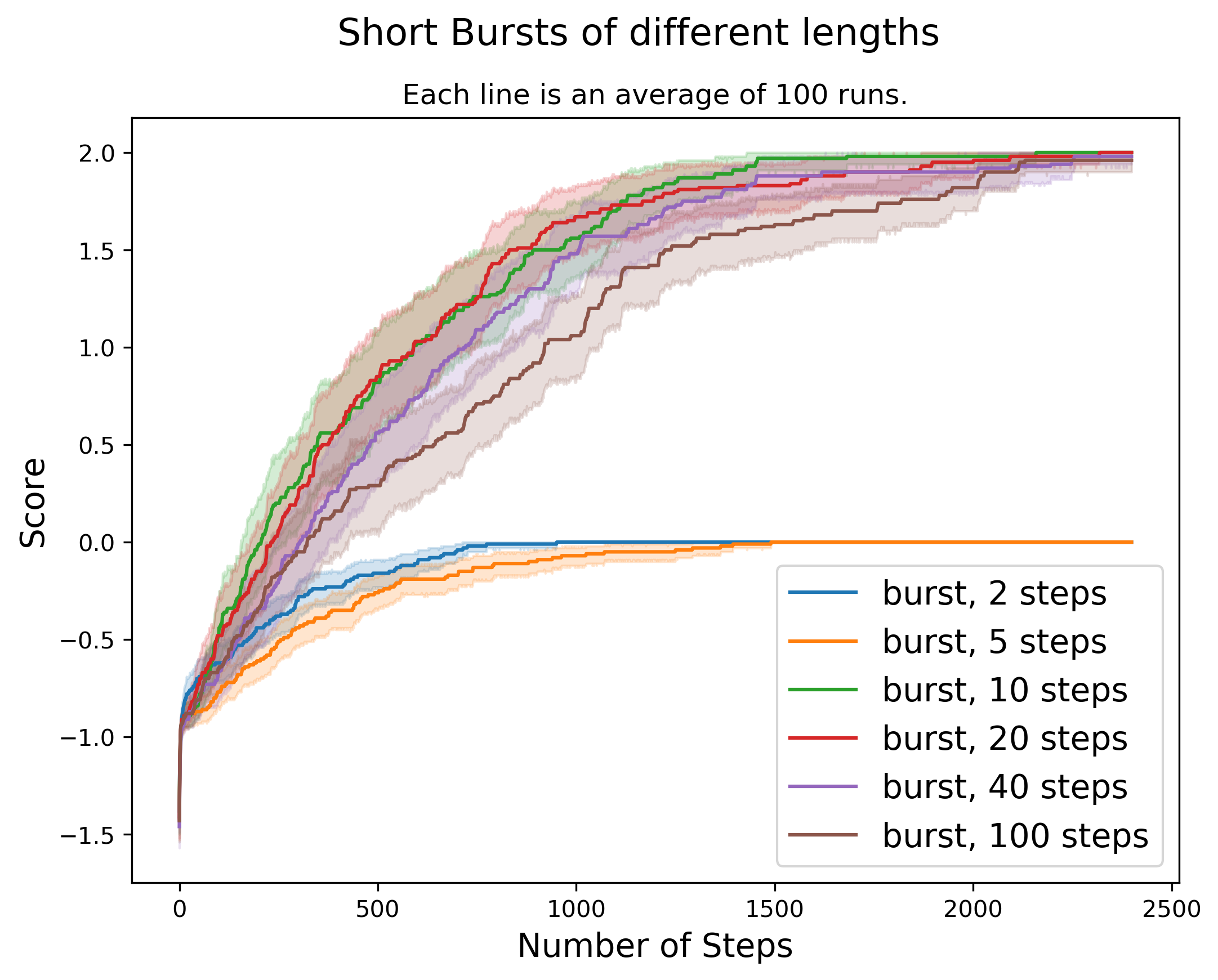}
        \caption{A comparison of different burst lengths}
        \label{fig:bottleneck_bs}
   \end{subfigure}
    
    \caption{Comparative behavior of short bursts and other optimization approaches on the $K_{20}, P_{7}$ bottleneck graph.  Each simulation was run 100 times, the lines in the plots show the average maximum values over time of those runs and the colored bars indicate a 95\% confidence interval for those averages.}
    \label{fig:bottleneck_results}
\end{figure}

\section{Conclusion}

We close with a summary of our findings, a discussion of limitations,  and ideas for future work. We have introduced an alternative to biased random walks for exploring the state space of a Markov chain to find states that are \oes with regard to some summary statistic. We have shown that this method, called {\em short bursts}, outperforms both unbiased and biased random walks in a practical use case: finding political districting plans with a large number of majority-minority districts. In this application setting, a wide range of values for the burst length parameter leads to an improvement over biased random walks, suggesting the success of the method is robust with regard to the choice of the burst length parameter. Finally, we explored simpler models, with both rigorous and empirical methods, to find possible explanations for the strong performance of short bursts in our empirical setting. 

We expect our results to generalize for any state and any minority population that is large enough and compact enough to comprise a majority in several districts. When a minority population is small enough and/or dispersed enough that it is only possible to comprise a minority in a few districts, we see less benefit of short bursts compared to biased random walks.
This observed with the black population in Texas, where BVAP only comprises 11.4\% of VAP and only 6 possible majority-BVAP districts were found (Figure~\ref{fig:tx-maxes-all}, d-f) and the  black population in Virginia, where BVAP only comprises 18.5\% of the population and only 12 majority-BVAP districts were found (Figure~\ref{fig:va-maxes-all}, a-c). In both cases, while short bursts seem to perform slightly better than biased random walks, significant overlap in the results of both methods gives us less confidence in our conclusions; further experiments would be needed to distinguish the two. 
This is the only effect we observe minority population having on our conclusions about short bursts, though it certainly may be interacting with various parameters in more subtle ways.

This work has some limitations that are worth noting. In particular, we do not yet fully understand why short bursts succeed to find political districting plans with many majority-minority districts. Though we propose one such explanation in Section \ref{sec:bottleneck}, we do not have evidence to verify it is correct. We also do not provide evidence as to the best possible choice of burst length.  Because many burst lengths perform similarly well in our experiments, more experiments beyond the proof-of-concept ones we present here would be needed to determine a statistically significant best choice of optimal burst length. More experiments could also shed light on which parameters of the problem (such as which percent of the population the minority group comprises or how the minority population is distributed throughout the state) affect the optimal burst length. However, we have observed no patterns so far and are doubtful whether any would emerge, provided the minority population is large enough to comprise a majority in a moderate number of districts. Finally, we do not consider how the method of optimizing a certain statistic via short bursts performs in other statistical problems outside the setting of political districting.

One direction of future work that will likely help us understand why short bursts work well for this problem would be to study the structure of the state space of political districting plans. Learning how districting plans are connected to each other by the moves of the Recombination Markov chain and what the energy landscape looks like will be helpful for identifying local maxima, bottlenecks, and other structural features that aid or hinder explorations by biased random walks and short bursts. 
Additional future work within the context of political districting includes studying how short bursts perform when optimizing other characteristics of districting plans: while we focus on the number of majority-minority districts, there are many other relevant statistics. For example, short bursts could be used to explore districting plans while trying to minimize the number of counties split between two districts, the compactness or competitiveness of districts, or the partisan balance of the state. It is not clear that short bursts would perform similarly well on these statistics. Exploration in this direction could also help explain why short bursts perform as well as they do for maximizing majority-minority districts and consequently how they might perform in settings unrelated to political districting.



\bibliography{biblio}

\newpage
\appendix
\section{Further Empirical Investigations} \label{app}
\normalsize

We also looked at additional states beyond Louisiana and Texas to understand the robustness of our findings. In Virginia, just like Louisiana, we looked for state House plans with many districts with majority Black voting age population (BVAP) as well as districts with a majority People of Color voting age population (POCVAP), see Figure~\ref{fig:fig:va-maxes-all}. In New Mexico, we looked for state House plans with many districts with majority Hispanic voting age population (HVAP) and POCVAP, see Figure~\ref{fig:fig:nm-maxes-all}. In our experiments on these states we used different population balances, in particular the tightest bound for which we could find a starting seed plan. For Virginia this was 2\%, and for New Mexico this was 4.5\%. Otherwise, we obtained data from similar governmental sources and ran the same experiments as for Louisiana. 

For the results of our experiment with Virginia, note there are two lines displaying enacted plans. This is because a new plan was enacted in 2019, after the 2011 Virginia House of Delegates map was found to be a racial gerrymander~\citep{bethunehill,vahousedelegates}. The 2019 plan has fewer majority-black districts, but is believed to give the black population a better opportunity to elect representatives of choice; see~\citep{virginia}. 

 Overall, these states show a similar story as seen in Louisiana, where short bursts generally outperform biased runs. However, when the range of possible values is narrow (see majority BVAP districts in VA, Figure~\ref{fig:va-maxes-all}, top), there is too much overlap in the expected maximum curves for the advantage of short bursts to be conclusive. This is the same effect observed for BVAP in Texas, see Figure~\ref{fig:tx-maxes-all}, (d)-(f). This narrow range of possible numbers of majority-minority districts is likely a combined effect of the spatial distribution of the population and the population of interest being a smaller fraction of the total population. In Virginia the BVAP accounts for 18.5\% of the total voting age population and in Texas BVAP is only 11.4\% of the total voting age population. In the rest of the plots, the minority voting age population accounts for roughly 30\% or more of the total population, and the advantage of short bursts over biased runs is clear.

\begin{figure}
    \centering
    \begin{subfigure}{\textwidth}
        \begin{subfigure}{0.44\textwidth}
            \centering
            \includegraphics[width=\textwidth]{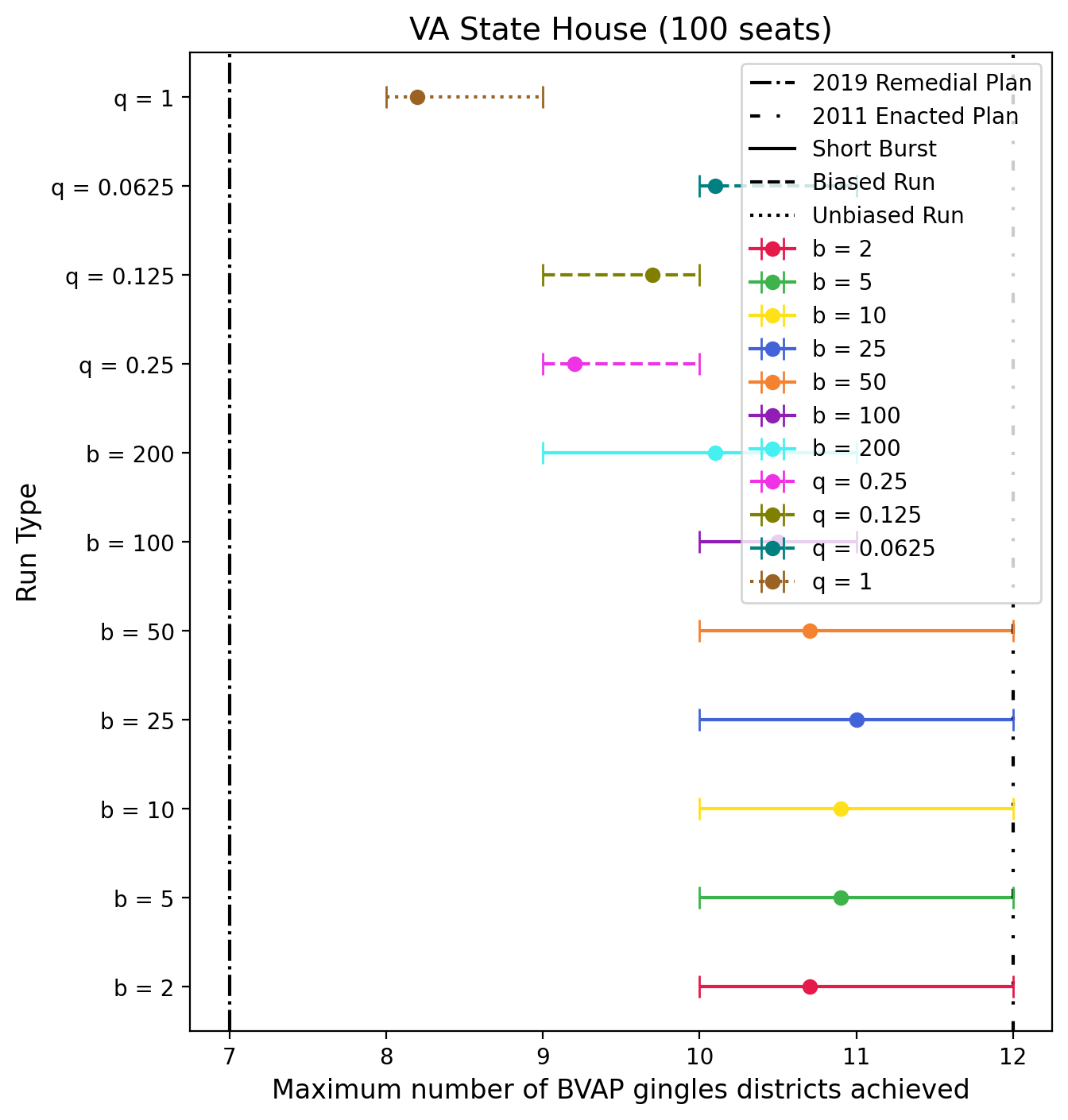}
            \caption{Maximum Gingles districts observed}
        \end{subfigure}
        \begin{subfigure}{0.55\textwidth}
            \centering
            \includegraphics[width=\textwidth]{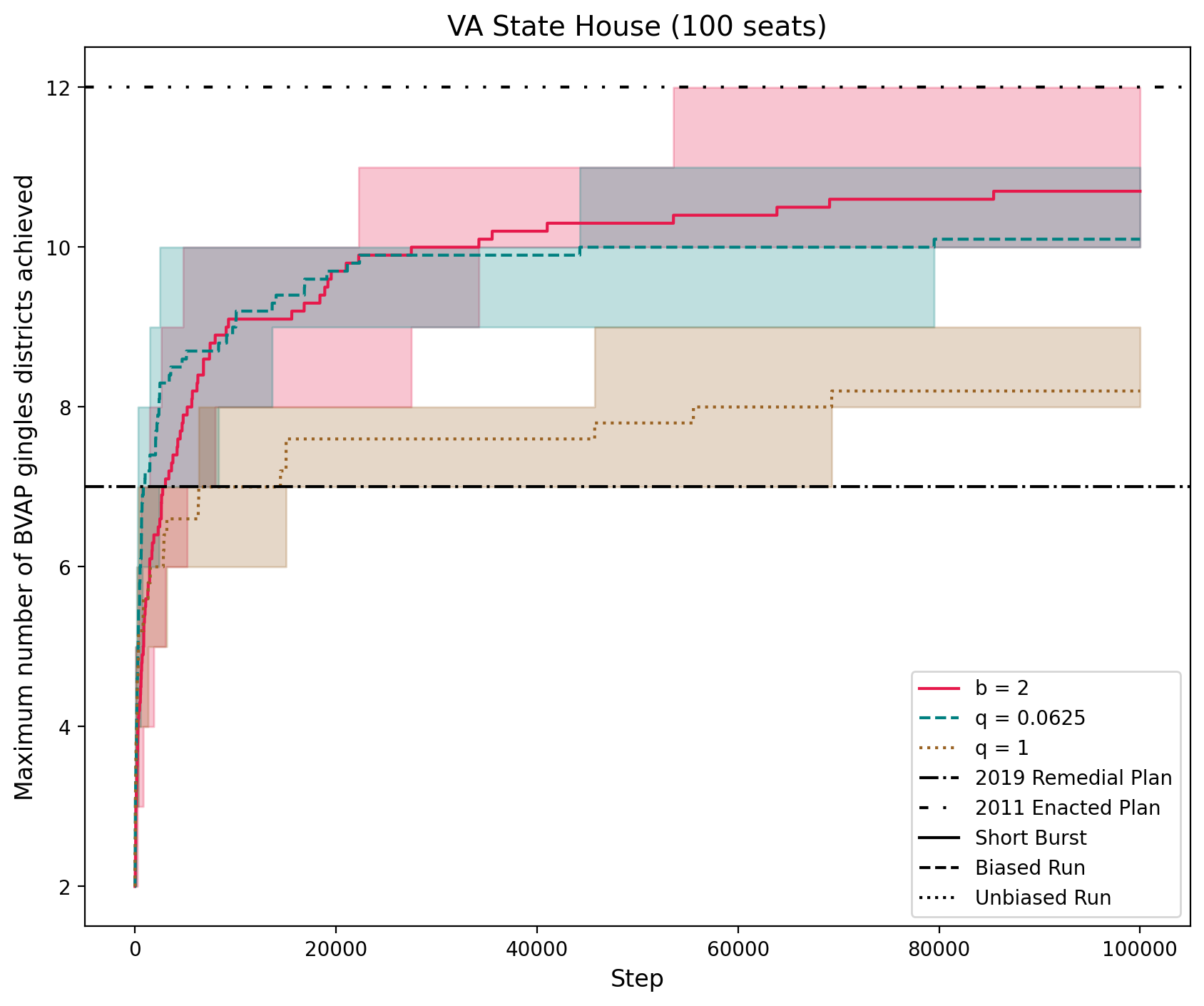}
            \caption{Maximums observed over time}
        \end{subfigure}
        \caption{BVAP}
    \end{subfigure}
    \begin{subfigure}{\textwidth}
        \begin{subfigure}{0.44\textwidth}
            \centering
            \includegraphics[width=\textwidth]{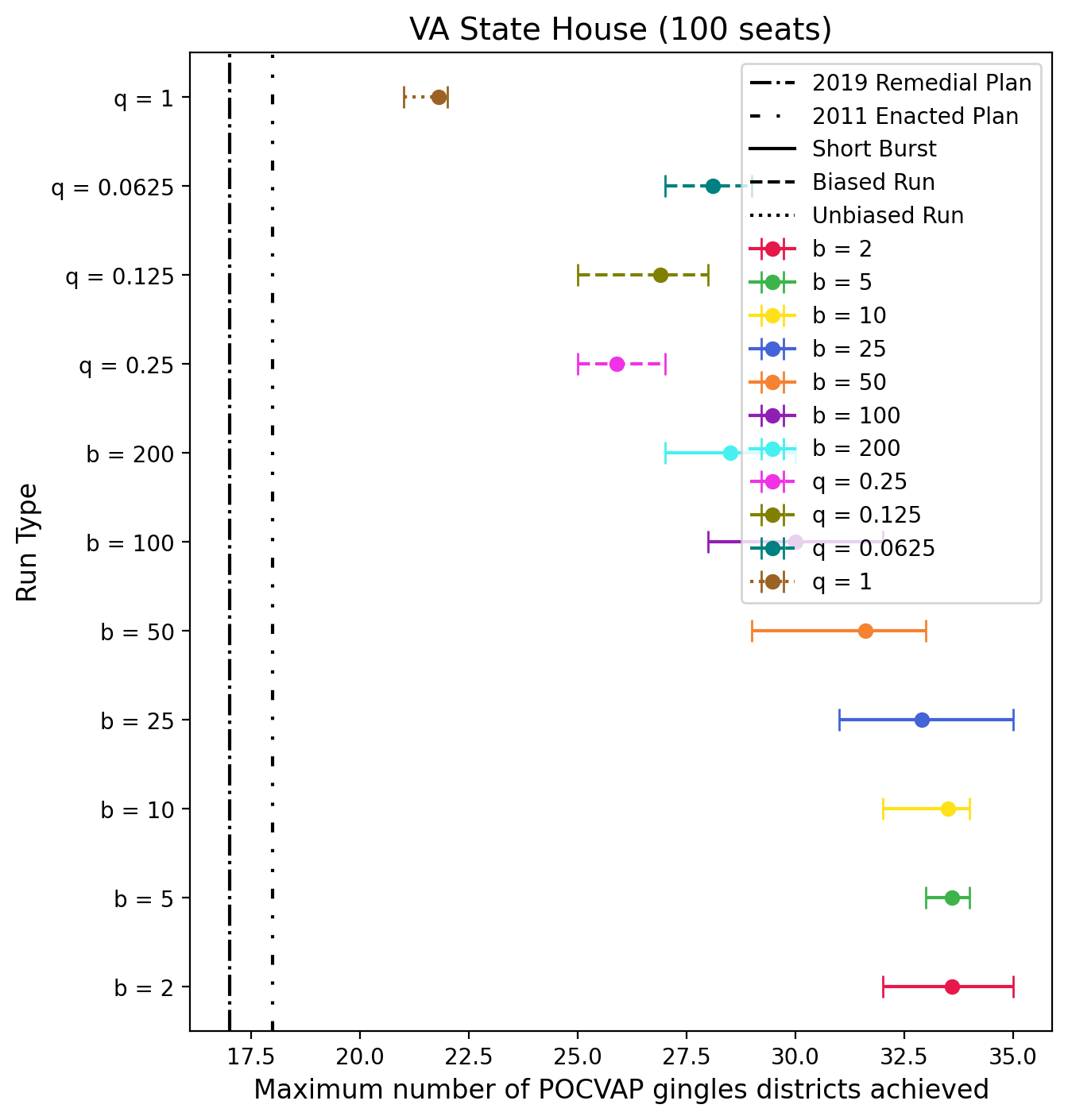}
            \caption{Maximum Gingles districts observed}
        \end{subfigure}
        \begin{subfigure}{0.55\textwidth}
            \centering
            \includegraphics[width=\textwidth]{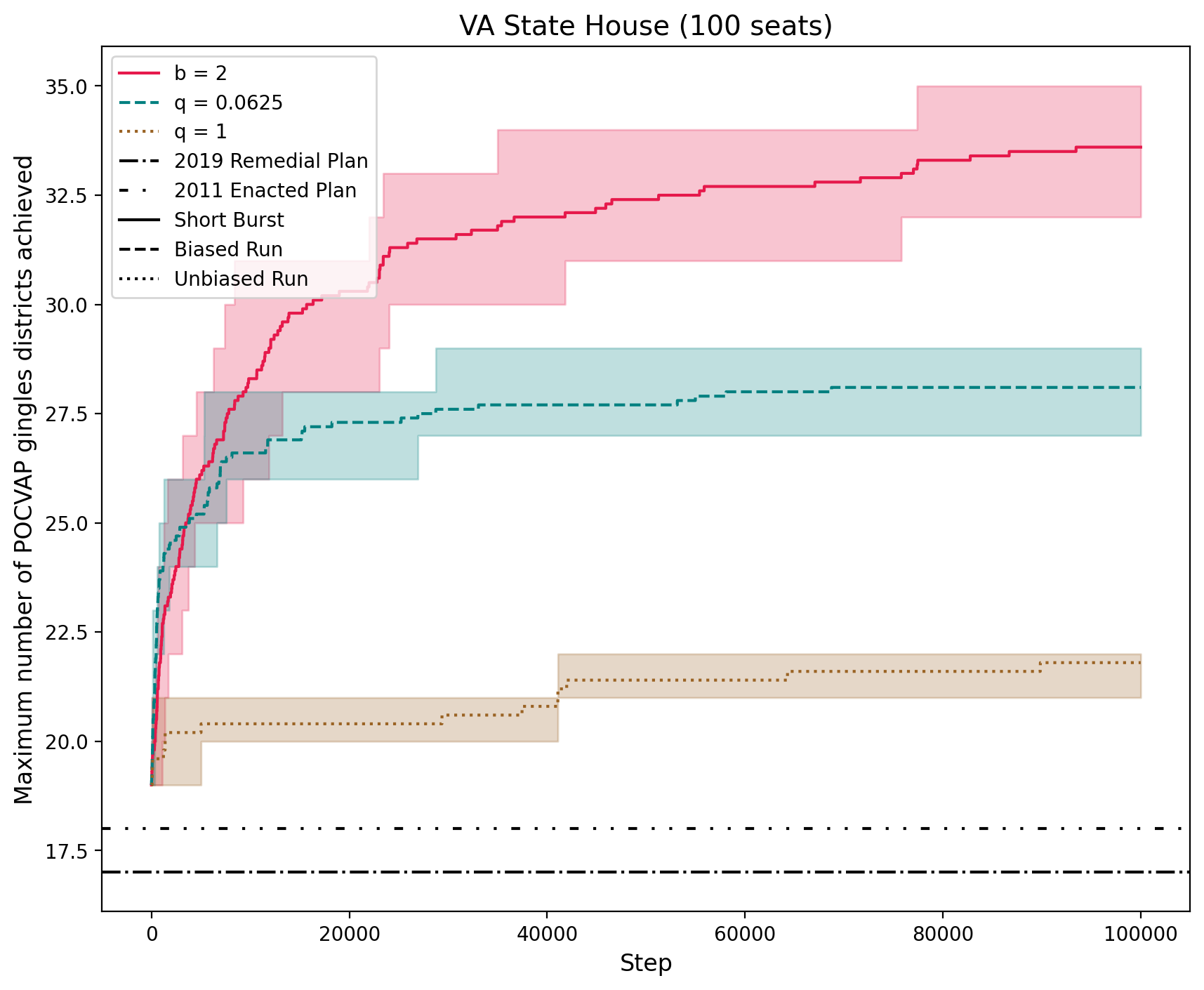}
            \caption{Maximums observed over time}
        \end{subfigure}
        \caption{POCVAP}
    \end{subfigure}

    \caption{Virginia | Maximum numbers of majority-minority districts observed for short bursts, biased runs, and an unbiased run ($q$ = 1). The short burst and biased 100,000 step runs were performed 10 times, and the unbiased run (q=1) was performed 5 times.  The left figures (a, d) show the range of maxes observed for each run type. The dot is the mean across the ten trials and the bars indicate the min/max range.  The right figures (b, e) show maximum number of majority-minority districts achieved at each step in the chain for the best preforming short burst length ($b$ = 5), biased run ($q$ = 0.0625), and the unbiased run ($q$ = 1).  The line indicates the mean, and the colored band the min/max range, across the trials.  The short bursts runs outperform the biased runs.
    }
    \label{fig:va-maxes-all}
\end{figure}

\begin{figure}
    \centering
    \begin{subfigure}{\textwidth}
        \begin{subfigure}{0.44\textwidth}
            \centering
            \includegraphics[width=\textwidth]{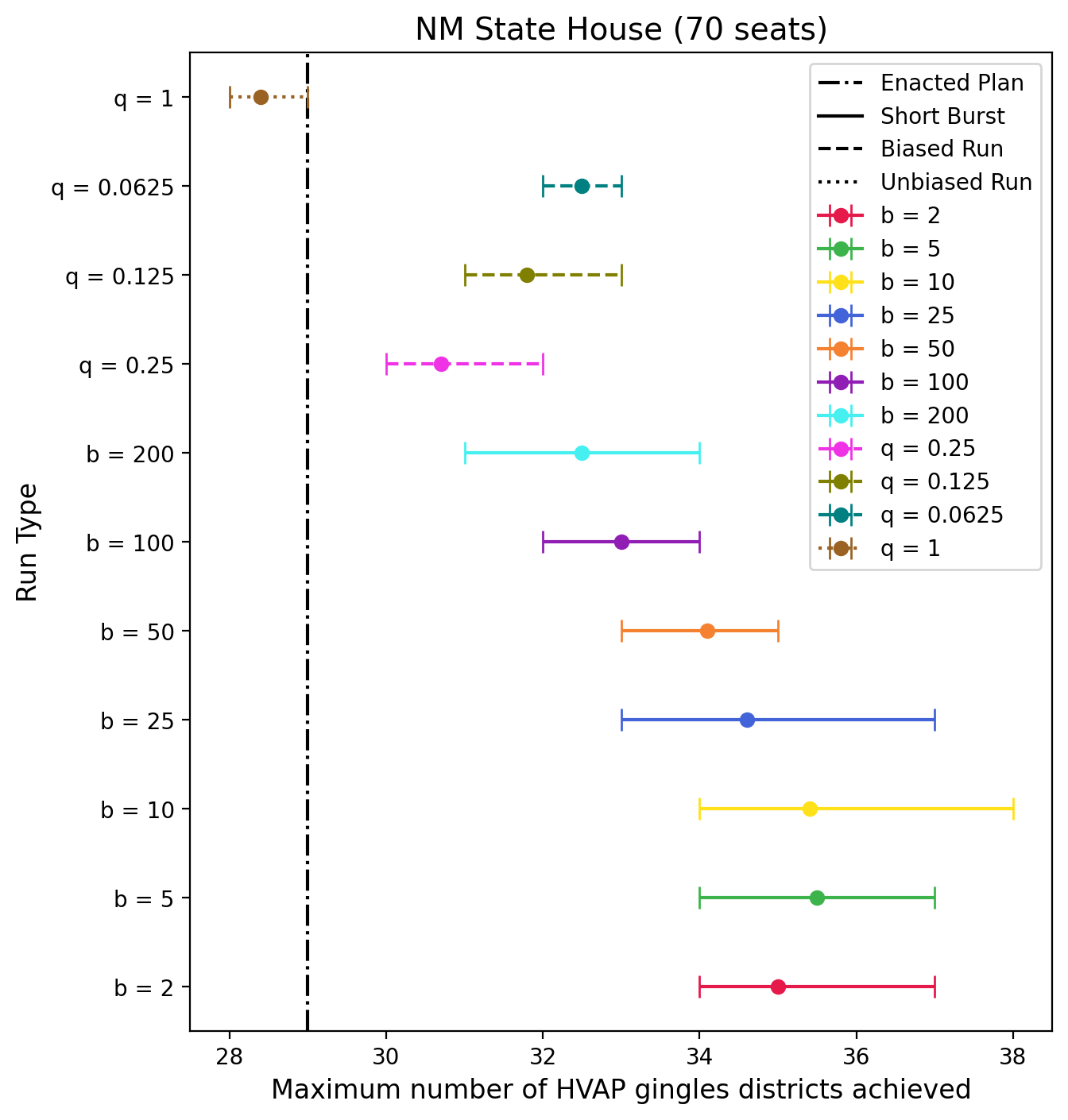}
            \caption{Maximum Gingles districts observed}
        \end{subfigure}
        \begin{subfigure}{0.55\textwidth}
            \centering
            \includegraphics[width=\textwidth]{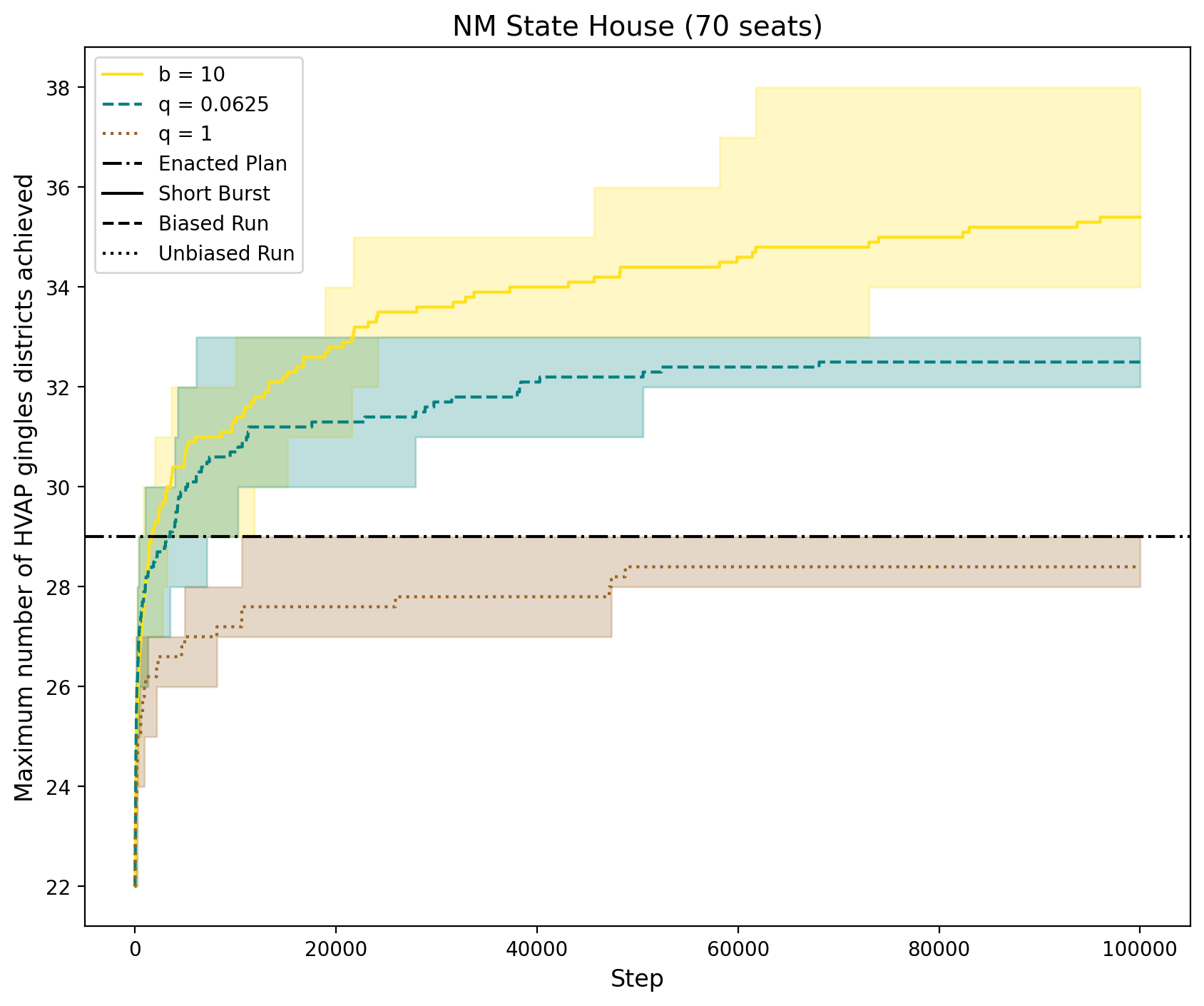}
            \caption{Maximums observed over time}
        \end{subfigure}
        \caption{HVAP}
    \end{subfigure}
    \begin{subfigure}{\textwidth}
        \begin{subfigure}{0.44\textwidth}
            \centering
            \includegraphics[width=\textwidth]{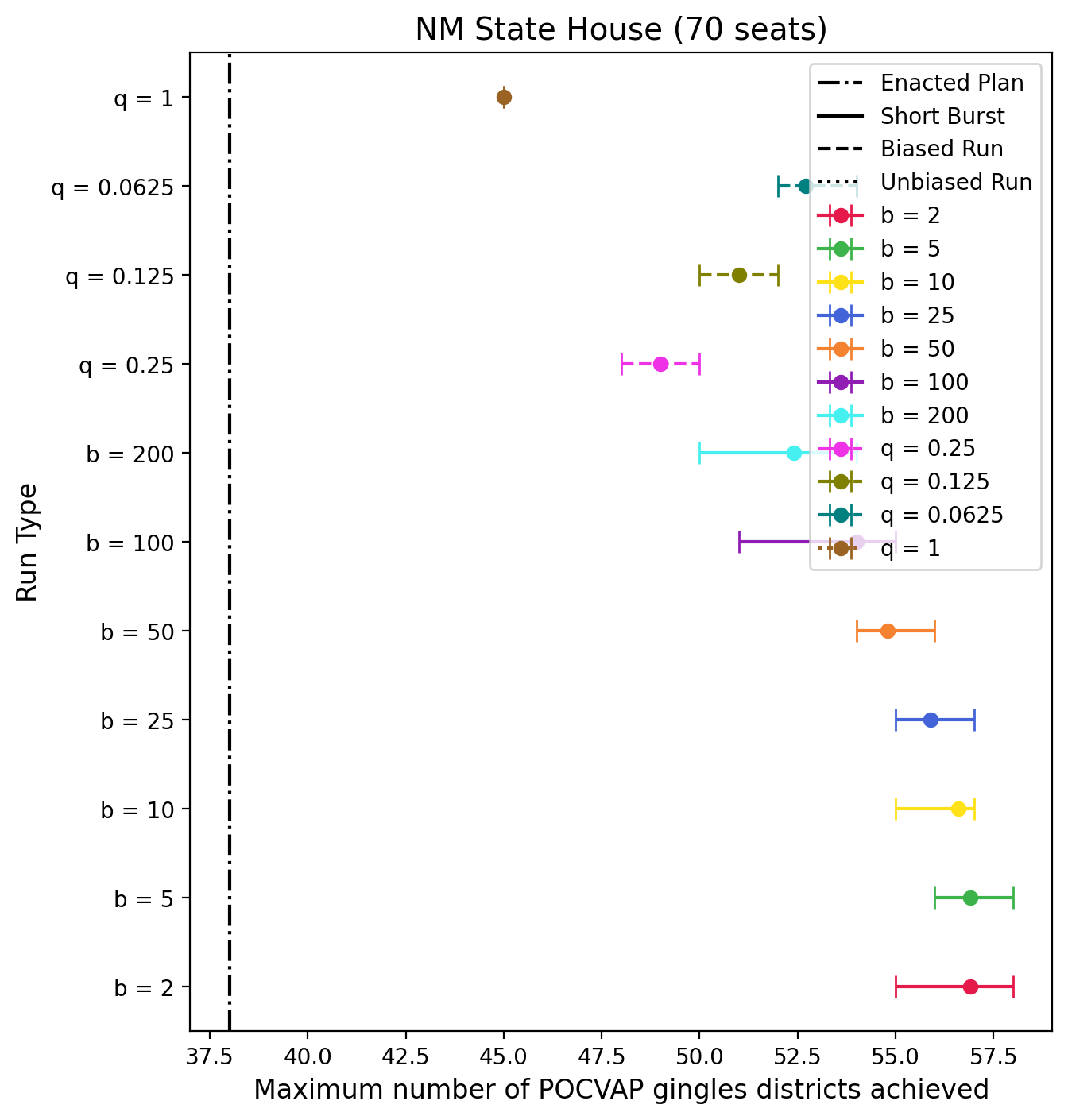}
            \caption{Maximum Gingles districts observed}
        \end{subfigure}
        \begin{subfigure}{0.55\textwidth}
            \centering
            \includegraphics[width=\textwidth]{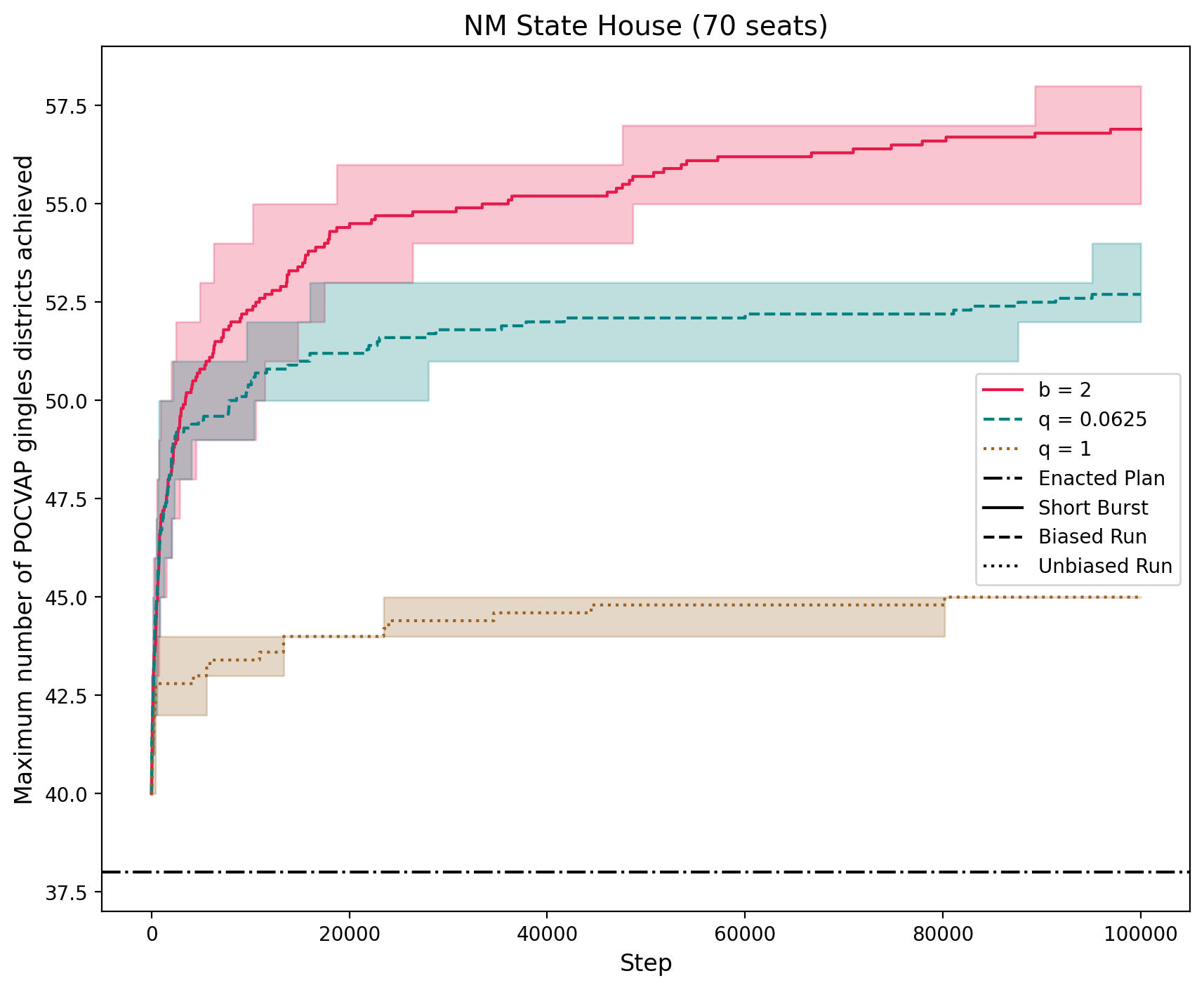}
            \caption{Maximums observed over time}
        \end{subfigure}
        \caption{POCVAP}
    \end{subfigure}

    \caption{New Mexico | Maximum numbers of majority-minority districts observed for short bursts, biased runs, and an unbiased run ($q$ = 1).  The short burst and biased 100,000 step runs were performed 10 times, and the unbiased run (q=1) was performed 5 times.  The left figures (a, d) show the range of maxes observed for each run type. The dot is the mean across the ten trials and the bars indicate the min/max range.  The right figures (b, e) show maximum number of majority-minority districts achieved at each step in the chain for the best preforming short burst length ($b$ = 5), biased run ($q$ = 0.0625), and the unbiased run ($q$ = 1).  The line indicates the mean, and the colored band the min/max range, across the trials.  The short bursts runs outperform the biased runs.
    }
    \label{fig:nm-maxes-all}
\end{figure}

    

\end{document}